\documentclass[11pt]{article}
\usepackage{amsmath,amssymb,amsthm,epsfig}
\usepackage[left=1in,top=1in,right=1in,bottom=1in,nohead]{geometry}
\usepackage{subfigure}
\usepackage{times}
\usepackage{url}
\begin{document}
\newtheorem{lemma}{Lemma}
\newtheorem{theorem}[lemma]{Theorem}
\newtheorem{informaltheorem}[lemma]{Informal Theorem}
\newtheorem{informallemma}[lemma]{Informal Lemma}
\newtheorem{corollary}[lemma]{Corollary}
\newtheorem{definition}[lemma]{Definition}
\newtheorem{proposition}[lemma]{Proposition}
\newtheorem{question}{Question}
\newtheorem{problem}{Problem}
\newtheorem{remark}[lemma]{Remark}
\newtheorem{claim}{Claim}
\newtheorem{fact}{Fact}
\newtheorem{challenge}{Challenge}
\newtheorem{observation}{Observation}
\newtheorem{openproblem}{Open Problem}
\newtheorem{openquestion}{Open question}
\newenvironment{LabeledProof}[1]{\bigskip \noindent{\bf Proof of #1:\\ }}{\qed}                        

\newenvironment{proofsketch}{\trivlist\item[]\emph{Proof Sketch}:}%
{\unskip\nobreak\hskip 1em plus 1fil\nobreak$\Box$
\parfillskip=0pt%
\endtrivlist}

\newcommand{\beq}{\begin{equation}}
\newcommand{\eeq}{\end{equation}}
\newcommand{\beas}{\begin{eqnarray*}}
\newcommand{\eeas}{\end{eqnarray*}}

\newcommand{\poly}{\mathrm{poly}}
\newcommand{\eps}{\epsilon}
\newcommand{\e}{\epsilon}
\newcommand{\polylog}{\mathrm{polylog}}
\newcommand{\rob}[1]{\left( #1 \right)} 
\newcommand{\sqb}[1]{\left[ #1 \right]} 
\newcommand{\cub}[1]{\left\{ #1 \right\} } 
\newcommand{\rb}[1]{\left( #1 \right)} 
\newcommand{\abs}[1]{\left| #1 \right|} 
\newcommand{\zo}{\{0, 1\}}
\newcommand{\zonzo}{\zo^n \to \zo}
\newcommand{\zokzo}{\zo^k \to \zo}
\newcommand{\zot}{\{0,1,2\}}

\newcommand{\en}[1]{\marginpar{\textbf{#1}}}
\newcommand{\efn}[1]{\footnote{\textbf{#1}}}

\newcommand{\junk}[1]{}
\newcommand{\prob}[1]{\Pr\left[ #1 \right]}
\newcommand{\expt}[1]{\mbox{E}\sqb{#1}}
\newcommand{\expect}[1]{\mbox{E}\left[ #1 \right]}

\newcommand{\md}[1]{\delta_{#1}}                            
\newcommand{\gf}[1]{G_{#1}}                                     
\newcommand{\nei}[3]{N^{#2}_{#1}\rb{#3}}                
\newcommand{\dg}[2]{d_{#1}\rb{#2}}                          
\newcommand{\dgi}[3]{d_{#1}\rb{#2,#3}}                   
\newcommand{\prb}[2]{p_{#1,#2}}

\newcommand{\BfPara}[1]{\noindent {\bf #1}.}


\begin{titlepage}

\title{Discovery through Gossip}

\author{Bernhard Haeupler \thanks{Computer Science and Artificial
    Intelligence Lab, Massachusetts Institute of Technology, Cambridge, MA 02139, USA. E-mail: {\tt haeupler@mit.edu}}
\and Gopal Pandurangan \thanks{Division of Mathematical
Sciences, Nanyang Technological University, Singapore 637371 and Department of Computer Science, Brown University, Providence, RI 02912, USA.  \hbox{E-mail}:~{\tt gopalpandurangan@gmail.com}. Supported in part by the following grants: Nanyang Technological University grant M58110000, Singapore Ministry of Education (MOE) Academic Research Fund (AcRF) Tier 2 grant MOE2010-T2-2-082,
US NSF grant CCF-1023166, and a grant from the US-Israel Binational Science
Foundation (BSF).} \and David Peleg \thanks{Department of Computer Science and Applied Mathematics, The Weizmann
Institute of Science, Rehovot, 76100 Israel.
E-mail: {\tt david.peleg@weizmann.ac.il}.
Supported by a grant from the United States-Israel Binational Science
Foundation (BSF).}
  \and Rajmohan Rajaraman \thanks{College of Computer and Information Science,
  Northeastern University, Boston  MA 02115, USA.
E-mail: {\tt \{rraj,austin\}@ccs.neu.edu}.  Supported in part by NSF grant CNS-0915985.} \and  Zhifeng Sun~$^\S$}

\date{}
 
\maketitle

\thispagestyle{empty}
\begin{abstract}
We study randomized gossip-based processes in dynamic networks that are motivated by discovery processes in large-scale distributed networks like peer-to-peer or social networks.\smallskip

A well-studied problem in peer-to-peer networks is the resource discovery problem. There, the goal for nodes (hosts with IP addresses) is to discover the IP addresses of all other hosts. In social networks, nodes (people) discover new nodes through exchanging contacts with their neighbors (friends). In both cases the discovery of new nodes changes the underlying network - new edges are added to the network - and the process continues in the changed network. Rigorously analyzing such dynamic (stochastic) processes with a continuously self-changing topology remains a challenging problem with obvious applications.\smallskip

This paper studies and analyzes two natural gossip-based discovery processes. In the push process, each node repeatedly chooses two random neighbors and puts them in contact (i.e., ``pushes'' their mutual information to each other). In the pull discovery process, each node repeatedly requests or ``pulls'' a random contact from a random neighbor. Both processes are lightweight, local, and naturally robust due to their randomization.\smallskip

Our main result is an almost-tight analysis of the time taken for these two randomized processes to converge. We show that in any undirected $n$-node graph both processes take $O(n \log^2 n)$ rounds to connect every node to all other nodes with high probability, whereas $\Omega(n \log n)$ is a lower bound. In the directed case we give an $O(n^2 \log n)$ upper bound and an $\Omega(n^2)$ lower bound for strongly connected directed graphs. A key technical challenge that we overcome is the analysis of a randomized process that itself results in a constantly changing network which leads to complicated dependencies in every round.

 \end{abstract}

{\bf Keywords:} Random process, Resource discovery, Social network, Gossip-based algorithm,
Distributed algorithm, Probabilistic analysis

\end{titlepage}

\section{Introduction}
Many large-scale, real-world networks such as peer-to-peer networks,
the Web, and social networks are highly dynamic with continuously
changing topologies. The evolution of the network as a whole is
typically determined by the decentralized behavior of nodes, i.e., the
local topological changes made by the individual nodes (e.g., adding
edges between neighbors).  Understanding the dynamics of such local
processes is critical for both analyzing the underlying stochastic
phenomena, e.g., in the emergence of structures in social networks,
the Web and other real-world networks \cite{b1,b2,b3}, and designing
practical algorithms for associated algorithmic problems, e.g., in
resource discovery in distributed networks \cite{leighton,law-siu} or
in the analysis of algorithms for the Web \cite{frieze1, frieze2}.  In
this paper, we study the dynamics of network evolution that result
from {\em local} gossip-style processes. Gossip-based processes have
recently received significant attention because of their simplicity of
implementation, scalability to large network size, and robustness to
frequent network topology changes; see, e.g., \cite{demers, kempe1,
  kempe2, chen-spaa, kempe, karp, shah, boyd, ozalp1, ozalp2} and the
references therein.  In particular, gossip-based protocols have been
used to efficiently and robustly construct various overlay topologies
dynamically in a fully decentralized manner \cite{ozalp1}.  In a local
gossip-based algorithm (e.g., \cite{chen-spaa}), each node exchanges
information with a small number of randomly chosen neighbors in each
round.\footnote{Gossip, in some contexts (see e.g.,
  \cite{karp,kempe}), has been used to denote communication with a
  random node in the network, as opposed to only a directly connected
  neighbor.  The former model essentially assumes that the underlying
  graph is complete, whereas the latter (as assumed here) is more
  general and applies even to arbitrary graphs. The local gossip
  process is typically more difficult to analyze due to the
  dependences that arise as the network evolves.}  The randomness
inherent in the gossip-based protocols naturally provides robustness,
simplicity, and scalability. While many of the recent theoretical
gossip-based work (including those on rumor spreading), especially,
the {\em push-pull} type algorithms (\cite{karp, kempe, chen-spaa,
  doerr, flavio, giakkoupis}) focus on analyzing various gossip-based
tasks (e.g., computing aggregates or spreading a rumor) on {\em
  static} graphs, a key feature of this work is rigorously analyzing a
gossip-based process in a {\em dynamically changing} graph.
     
We present two illustrative application domains for our study.  First,
consider a P2P network, where nodes (computers or end-hosts with
IDs/IP addresses) can communicate only with nodes whose IP address are
known to them.  A basic building block of such a dynamic distributed
network is to efficiently discover the IP addresses of all nodes that
currently exist in the network.  This task, called {\em resource
  discovery} \cite{leighton}, is a vital mechanism in a dynamic
distributed network with many applications~\cite{leighton,ittai}: when
many nodes in the system want to interact and cooperate they need a
mechanism to discover the existence of one another.  Resource
discovery is typically done using a local mechanism \cite{leighton};
in each {\em round}\/ nodes discover other nodes and this changes the
resulting network --- new edges are added between the nodes that
discovered each other.  As the process proceeds, the graph becomes
denser and denser and will finally result in a complete graph.  Such a
process was first studied in \cite{leighton} which showed that a
simple randomized process is enough to guarantee almost-optimal time
bounds for the time taken for the entire graph to become complete
(i.e., for all nodes to discover all other nodes). Their randomized
{\em Name Dropper} algorithm operates as follows: in each round, each
node chooses a random neighbor and sends {\em all} the IP addresses it
knows.  Note that while this process is also gossip-based the
information sent by a node to its neighbor can be extremely large
(i.e., of size $\Omega(n)$).  More recently, self-stabilization
protocols have been designed for constructing and maintaining P2P
overlay networks e.g, \cite{berns,jacob}. These protocols guarantee
convergence to a desired overlay topology (e.g., the SKIP+ graph)
starting from any arbitrary topology via local checking and repair.
For example, the self-stabilizing protocol of \cite{berns} proceeds by
continuously discovering new neighbors (via transitive closure) till a
complete graph is formed. Then the repair process is initiated. This
can also be considered as a local gossip-based process in an
underlying virtual graph with changing (added) edges. \junk{However,
  the process is not lightweight as information sent by a node to its
  neighbor can be extremely large (i.e., of size $\Omega(n)$).}  In
both the above examples, the assumption is that the starting graph is
arbitrary but (at least) weakly connected.  The gossip-based processes
that we study also have the same goal --- starting from an arbitrary
connected graph, each node discovers all nodes as quickly as possible
-- in a setting where individual message sizes are small ($O(\log n)$
bits).

Second, in social networks, nodes (people) discover new nodes through
exchanging contacts with their neighbors (friends). Discovery of new
nodes changes the underlying network --- new edges are added to the
network --- and the process continues in the changed network.  For
example, consider the {\em LinkedIn}
network\footnote{\url{http://www.linkedin.com}.}, a large social
network of professionals on the Web. The nodes of the network
represent people and edges are added between people who directly know
each other --- between direct contacts.  Edges are generally
undirected, but LinkedIn also allows directed edges, where only one
node is in the contact list of another node.  LinkedIn allows two
mechanisms to discover new contacts.  The first can be thought of as a
{\em triangulation} process (see Figure~\ref{fig:intro}(a)): A person
can introduce two of his friends that could benefit from knowing each
other --- he can mutually introduce them by giving their contacts. The
second can be thought of as a {\em two-hop} process (see
Figure~\ref{fig:intro}(b)): If {\em you} want to acquire a new contact
then you can use a shared (mutual) neighbor to introduce yourself to
this contact; i.e., the new contact has to be a two-hop neighbor of
yours.  Both the processes can be modeled via gossip in a natural way
(as we do shortly below) and the resulting evolution of the network
can be studied: e.g., how and when do clusters emerge?  how does the
diameter change with time?  In the social network context, our study
focuses on the following question: how long does it take for all the
nodes in a connected induced subgraph of the network to discover all
the nodes in the subgraph?  This is useful in scenarios where members
of a social group, e.g., alumni of a school, members of a club,
discover all members of the group through local gossip operations.

\begin{figure}[ht]
\begin{center}
  \includegraphics[width=6in]{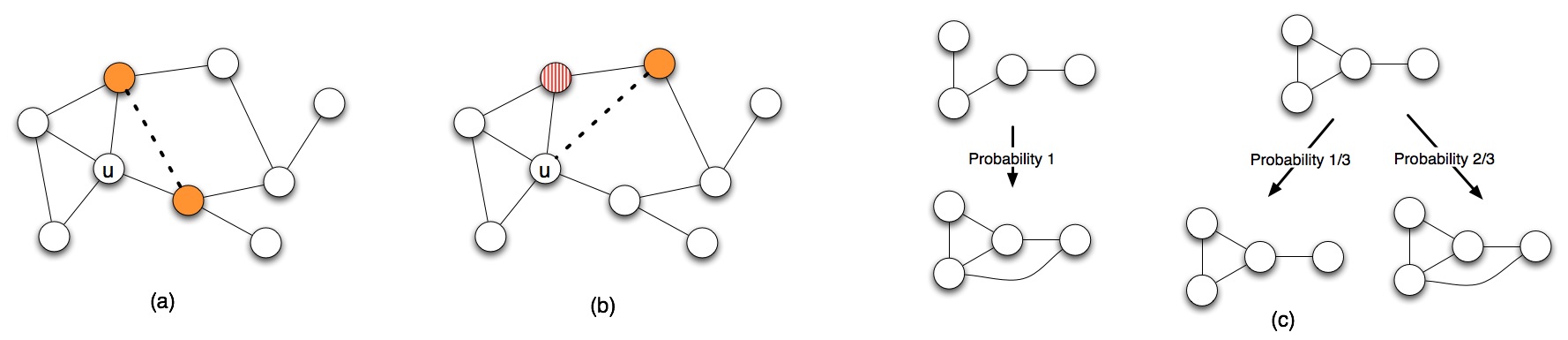}
 \caption{(a) Push discovery or triangulation process. (b) Pull
   discovery or two-hop walk process. (c) Non-monotonicity of the
   triangulation process -- the expected convergence time for the
   4-edge graph exceeds that for the 3-edge
   subgraph.\label{fig:intro}}
\end{center}
\end{figure}

\BfPara{Gossip-based discovery}  Motivated directly by the above
applications, we analyze two lightweight, randomized gossip-based
discovery processes.  We assume that we start with an arbitrary
undirected connected graph and the process proceeds in synchronous
rounds.  Communication among nodes occurs only through edges in the
network. We further assume that the size of each message sent by a
node in a round is at most $O(\log n)$ bits, i.e., the size of an ID.
  \begin{enumerate}
\item {\sf Push discovery (triangulation)}: In each round, each
  node chooses two random neighbors and connects them by ``pushing''
  their mutual information to each other. In other words, each node
  adds an undirected edge between two of its random neighbors; if the
  two neighbors are already connected, then this does not create any
  new edge.  Note that this process, which is illustrated in
  Figure~\ref{fig:intro}(a), is completely local.  To execute the
  process, a node only needs to know its neighbors; in particular, no
  two-hop information is needed. Note that this is similar in spirit to the {\em triangulation} procedure of Linkedin described earlier, i.e., a node completes a triangle with two of its chosen neighbors. \footnote{However, we note that in our process the two neighbors are chosen randomly, unlike in LinkedIn.}  

 \item {\sf Pull discovery (two-hop walk)}: In each round, each node
   connects itself to a random neighbor of a neighbor chosen uniformly
   at random, by ``pulling'' a random neighboring ID from a random
   neighbor.  Alternatively, one can think of each node doing a
   two-hop random walk and connecting to its destination.  This
   process, illustrated in Figure~\ref{fig:intro}(b), can also be
   executed locally: a node simply asks one of its neighbors $v$ for
   an ID of one of $v$'s neighbors and then adds an undirected edge to
   the received contact.  Note that this is similar in spirit to the
   {\em two-hop} procedure of LinkedIn described earlier
   \footnote{Again, one difference is that in the process we analyze
     the particular each node in the two-hop walk is chosen uniformly
     at random from the appropriate neighborhood.}.
 \end{enumerate}
  
  Both the above processes are local in the sense that each node only
  communicates with its neighbors in any round, and lightweight in the
  sense that the amortized work done per node is only a constant per
  round.  Both processes are also easy to implement and generally
  oblivious to the current topology structure, changes or failures.
  It is interesting also to consider variants of the above processes
  in directed graphs. In particular, we study the two-hop walk process
  which naturally generalizes in directed graphs: each node does a
  two-hop directed random walk and adds a {\em directed}\/ edge to its
  destination.  We are mainly interested in the time taken by the
  process to converge to the {\em transitive closure} of the initial
  graph, i.e., till no more new edges can be added.  \junk{In an
    undirected graph, the processes will converge to a complete graph,
    while that may not necessarily be the case in directed graphs.}
  

\smallskip  
\BfPara{Our results}   
Our main contribution is an analysis of the above gossip-based
discovery processes in both undirected and directed graphs.  In
particular, we show the following results (the precise theorem
statements are in the respective sections.)

\begin{itemize}
\item {\bf Undirected graphs:} In Sections~\ref{sec:triangulation-10p}
  and \ref{sec:2hop-10p}, we show that for {\em any} undirected
  $n$-node graph, both the push and the pull discovery processes
  converge in $O(n\log^2 n)$ rounds with high probability.  We also
  show that $\Omega(n \log n)$ is a lower bound on the number of
  rounds needed for almost any $n$-node graph. Hence our analysis is
  tight to within a logarithmic factor.  Our results also apply when
  we require only a subset of nodes to converge.  In particular,
  consider a subset of $k$ nodes that induce a connected subgraph and
  run the gossip-based process {\em restricted to this subgraph}.
  Then by just applying our results to this subgraph, we immediately
  obtain that it will take $O(k\log^2 k)$ rounds, with high
  probability (in terms of $k$), for all the nodes in the subset to
  converge to a complete subgraph.  As discussed above, such a result
  is applicable in social network scenarios where all nodes in a
  subset of network nodes discover one another through gossip-based
  processes.
    
 \item {\bf Directed graphs:} In Section \ref{sec:directed-10p}, we show
   that the pull process takes $O(n^2 \log n)$ time for any $n$-node
   directed graph, with high probability.  We show a matching lower
   bound for weakly connected graphs, and an $\Omega(n^2)$ lower bound
   for strongly connected directed graphs.  Our analysis indicates
   that the directionality of edges can greatly impede the resource
   discovery process.  \junk{
  \item Can we talk about the message (communication) complexity 
  and the bit complexity of our algorithms (e.g., these are done
  in prior works in resource discovery, see e.g., the paper by Abraham
  and Dolev --- available in our kdissemination website.)
   
  \item Other results to add ? --- e.g., robustness to failures, only subset of nodes participating etc.
}
\end{itemize}  

\BfPara{Applications} 
The gossip-based discovery processes we study are directly motivated
by the two scenarios outlined above, namely algorithms for resource
discovery in distributed networks and analyzing how discovery process
affects the evolution of social networks. Since our processes are
simple, lightweight, and easy to implement, they can be used for
resource discovery in distributed networks.  The {\em Name Dropper}\/
discovery algorithm has been applied to content delivery
systems\cite{leighton}.  As mentioned earlier, {\em Name Dropper}\/
and other prior algorithms for the discovery problem \cite{leighton,
  law-siu, kutten, ittai} complete in polylogarithmic number of rounds
($O(\log^2 n)$ or $O(\log n)$), but may transfer $\Theta(n)$ bits per
edge per round.  As a result, they may not be scalable for bandwidth
and resource-constrained networks (e.g., peer-to-peer, mobile, or
sensor networks).  One approach to use these algorithms in a
bandwidth-limited setting ($O(\log n)$-bits per message) is to spread
the transfer of long messages over a linear number of rounds, but this
requires coordination and maintaining state.  In contrast, the
``stateless'' nature of the gossip processes we study and the fact
that the results apply to any initial graph make the process
attractive in unpredictable environments.  \junk{ In contrast, the
  {\em Name Dropper} algorithm of \cite{leighton}, . We note that,
  however, because there is essentially no restriction on the
  bandwidth, the number of rounds taken by the {\em Name Dropper}
  algorithm is $O(\log^2 n)$. (We note that in our model, $\Omega(n)$
  is a trivial lower bound).} Our analyses can also give insight into
the growth of real-social networks such as LinkedIn, Twitter, or
Facebook, that grow in a decentralized way by the local actions of the
individual nodes.  In addition to the application of discovering all
members of a group, analyses of the processes such as the ones we
study can help analyze both short-term and long-term evolution of
social networks.  In particular, it can help in predicting the sizes
of the immediate neighbors as well as the sizes of the second and
third-degree neighbors (these are listed for every node in LinkedIn).
An estimate of these can help in designing efficient algorithms and
data structures to search and navigate the social network.

\smallskip
\BfPara{Technical contributions} Our main technical contribution is a
probabilistic analysis of localized gossip-based discovery in
arbitrary networks.  While our processes can be viewed as graph-based
coupon collection processes, one significant distinction with past
work in this
area~\cite{adler+hkv:p2p,alon:combinatorics,dimitriov+p:coupon} is
that the graphs in our processes are constantly changing.  The
dynamics and locality inherent in our process introduces nontrivial
dependencies, which makes it difficult to characterize the network as
it evolves.  A further challenge is posed by the fact that the
expected convergence time for the two processes is {\em not
  monotonic}; that is, the processes may {\em take longer}\/ to
converge starting from a graph $G$ than starting from a subgraph $H$
of $G$.  Figure~\ref{fig:intro}(c) presents a small example
illustrating this phenomenon.  This seemingly counterintuitive
phenomenon is, however, not surprising considering the fact that the
cover time of random walks also share a similar property.  One
consequence of these hurdles is that analyzing the convergence time
for even highly specialized or regular graphs is challenging since the
probability distributions of the intermediate graphs are hard to
specify.  Our lower bound analysis for a specific strongly connected
directed graph in Theorem~\ref{thm:directed.lower-10p} illustrates
some of the challenges.  In our main upper bound results
(Theorems~\ref{thm:triangulation-10p}
and~\ref{thm:graph+randwalk-10p}), we overcome these technical
difficulties by presenting a uniform analysis for all graphs, in which
we study different local neighborhood structures and show how each
leads to rapid growth in the minimum degree of the graph.

\junk{
\paragraph{Other related work.} ?

\paragraph{Organization of the paper.} Giving a road map of the sections here
can be useful...
 }

\section{Preliminaries}
\label{sec:prelim}
In this section, we define the notations used in our proofs, and prove
some common lemmas for Section \ref{sec:triangulation-10p} and Section
\ref{sec:2hop-10p}.  Let $G$ denote a connected graph, $d(u)$ denote
the degree of node $u$, and $N^i(u)$ denote the set of nodes that are
at distance $i$ from $u$. Let $\delta$ denote the minimum degree of
$G$. We note that $G$, $d(u)$, and $N^i(u)$ all change with time, and
are, in fact, random variables. For any nonnegative integer $t$, we
use subscript $t$ to denote the random variable at the start of round
$t$; for example $\gf{t}$ refers to the graph at the start of round
$t$. For convenience, we list the notations in
Table~\ref{tab:notation-10p}.

\begin{table}[htp]
\caption{Notation table \label{tab:notation-10p}}
\begin{center}
\begin{tabular}{|l|l|}
\hline
Notation & description \\
\hline
$\md{t}$ & minimum degree of graph $\gf{t}$ \\
$\nei{t}{i}{u}$ & set of nodes that are at distance $i$ from $u$ in
$\gf{t}$ \\
$\abs{\nei{t}{i}{u}}$ & number of nodes in $\nei{t}{i}{u}$ \\
$\dg{t}{u}$ & degree of node $u$ in $\gf{t}$ \\
$\dgi{t}{u}{\nei{t}{i}{v}}$ & number of edges from $u$ to nodes in
$\nei{t}{i}{v}$, i.e., degree induced on $\nei{t}{i}{v}$ \\
\hline
\end{tabular}
\end{center}
\end{table}

We state two lemmas that are used in the proofs in Section
\ref{sec:triangulation-10p} and Section \ref{sec:2hop-10p}. Lemma
\ref{lem:moreneighbor-10p} gives a lower bound on the number of
neighbors within distance 4 for any node $u$ in $\gf{t}$ while Lemma
\ref{lem:couponcorollary-10p} is a standard analysis of a sequence of
Bernoulli experiments and can be proved by a direct coupon collector
argument or using a Chernoff bound.  The proofs are included in
Appendix~\ref{app:prelim} for completeness.

\begin{lemma}
\label{lem:moreneighbor-10p}
$\abs{\cup_{i=1}^4 \nei{t}{i}{u}} \ge \min\cub{2\md{t},n-1}$ for all
$u$ in $\gf{t}$.
\end{lemma}

\junk{
\begin{proof}
  If $\nei{t}{3}{u}$ is not an empty set, consider node $v\in
  \nei{t}{3}{u}$. Since $\dg{t}{v} \ge \md{t}$, we have
  $\abs{\cup_{i=2}^4 \nei{t}{i}{u}} \ge
  \md{t}$. $\abs{\nei{t}{1}{u}}\ge \md{t}$ because $\dg{t}{u}\ge
  \md{t}$. We also know $\nei{t}{1}{u}$ and $\cup_{i=2}^4
  \nei{t}{i}{u}$ are disjoint. Thus, $\abs{\cup_{i=1}^4 \nei{t}{i}{u}}
  \ge 2\md{t}$.  If $\nei{t}{3}{u}$ is an empty set, then
  $\nei{t}{1}{u}\cup \nei{t}{2}{u}=n-1$ because $\gf{t}$ is
  connected. Thus $\abs{\cup_{i=1}^4 \nei{t}{i}{u}} = n-1$. Combine
  the above 2 cases, we complete the proof of this lemma.
\end{proof}
}

\begin{lemma}
\label{lem:couponcorollary-10p}
Consider $k$ Bernoulli experiments, in which the success probability
of the $i$th experiment is at least $i/m$ where $m \ge k$.  If $X_i$
denotes the number of trials needed for experiment $i$ to output a
success and $X=\sum_{i=1}^k X_i$, then $\prob{X>(c+1)n\ln n}$ is less
than $1/n^c$.
\end{lemma}

\junk{
\begin{proof}
  Since $X$ only increases with $k$, with out loss of generality
  assume that $k = m$. Now we can view this as {\em coupon collector
    problem} \cite{upfal} where $X_{m+1-i}$ is the number of steps to
  collect the $i$th coupon. Consider the probability of not obtaining
  the $i$th coupon after $(c+1)n\ln n$ steps. This probability is
\[\rb{1-\frac{1}{n}}^{(c+1)n\ln n} < e^{-(c+1)\ln n} =
\frac{1}{n^{c+1}}\] By union bound, the probability that some coupon
has not been collected after $(c+1)n\ln n$ steps is less than
$1/n^c$. And this completes the proof of this lemma.
\end{proof}
}

\section{Proofs for the triangulation process}
\label{sec:triangulation-10p}
In this section, we analyze the triangulation process on undirected
connected graphs, which is described by the following simple
iteration: In each round, for each node $u$, we add edge $(v,w)$
where $v$ and $w$ are drawn uniformly at random from
$\nei{t}{1}{u}$. The triangulation process yields the following
push-based resource discovery protocol. In each round, each node $u$
introduces two random neighbors $v$ and $w$ to one another.  The main
result of this section is that the triangulation process transforms an
arbitrary connected $n$-node graph to a complete graph in $O(n
\log^2 n)$ rounds with high probability.  We also establish an
$\Omega(n\log n)$ lower bound on the triangulation process for almost
all $n$-node graphs.

\subsection{Upper bound}
We obtain the $O(n \log^2 n)$ upper bound by proving that the minimum
degree of the graph increases by a constant factor (or equals $n-1$)
in $O(n \log n)$ steps.  Towards this objective, we study how the
neighbors of a given node connect to the two-hop neighbors of the
node.  We say that a node $v$ is {\bf {\em weakly tied}} to a set
of nodes $S$ if $v$ has less than $\md{0}/2$ edges to $S$
(i.e., $\dgi{t}{v}{S}<\md{0}/2$), and {\bf {\em strongly tied}} to $S$
if $v$ has at least $\md{0}/2$ edges to $S$ (i.e., $\dgi{t}{v}{S}
\ge\md{0}/2$).  (Recall that $\md{0}$ is the minimum degree at start of
round $0$.)
\begin{lemma}
\label{ob:strong-10p}
If $\md{0} \le \dg{t}{u} < (1+1/4)\md{0}$ and $w\in \nei{0}{1}{u}$ is
strongly tied to $\nei{t}{2}{u}$, then the probability that $u$
connects to a node in $\nei{t}{2}{u}$ through $w$ in round $t$ is at
least $2/(7n)$.
\junk{
\[\prob{u\mbox{ connects to a node in }\nei{t}{2}{u}\mbox{ through
  }w\mbox{ in round }t} \ge \frac{2}{7n} \]}
\end{lemma}
\begin{proof}
Since $w$ is strongly tied to $\nei{t}{2}{u}$,
$\dgi{t}{w}{\nei{t}{2}{u}} \ge \md{0}/2$.  Therefore, the probability
that $u$ connects to a node in $\nei{t}{2}{u}$ through $w$ in round
$t$ is
\begin{eqnarray*}
& = & \frac{\dgi{t}{w}{\nei{t}{2}{u}}}{\dg{t}{w}} \cdot
\frac{1}{\dg{t}{w}} 
\;\;\; \ge \;\;\; \frac{\dgi{t}{w}{\nei{t}{2}{u}}}{\dg{t}{w}}\cdot \frac{1}{n} 
\;\;\; \ge \;\;\; \frac{\dgi{t}{w}{\nei{t}{2}{u}}}{|\nei{t}{1}{u}|+\dgi{t}{w}{\nei{t}{2}{u}}}\cdot \frac{1}{n} \\
& \ge & 
\frac{\dgi{t}{w}{\nei{t}{2}{u}}}{(1+1/4)\md{0} +
  \dgi{t}{w}{\nei{t}{2}{u}}}\cdot \frac{1}{n} 
\;\;\; \ge \;\;\; \frac{\md{0}/2}{(1+1/4)\md{0} + \md{0}/2}\cdot \frac{1}{n} 
\;\;\; = \;\;\; \frac{2}{7n}.
\end{eqnarray*}
\junk{
\begin{eqnarray*}
& & \prob{u\mbox{ connects to a node in }\nei{t}{2}{u}\mbox{ through
  }w\mbox{ in round }t} \\
&=& \frac{\dgi{t}{w}{\nei{t}{2}{u}}}{\dg{t}{w}} \cdot
\frac{1}{\dg{t}{w}} \\
&\ge& \frac{\dgi{t}{w}{\nei{t}{2}{u}}}{\dg{t}{w}}\cdot \frac{1}{n} \\
&\ge& \frac{\dgi{t}{w}{\nei{t}{2}{u}}}{|\nei{t}{1}{u}|+\dgi{t}{w}{\nei{t}{2}{u}}}\cdot \frac{1}{n} \\
&\ge& \frac{\dgi{t}{w}{\nei{t}{2}{u}}}{(1+1/4)\md{0} +
  \dgi{t}{w}{\nei{t}{2}{u}}}\cdot \frac{1}{n} \\
&\ge& \frac{\md{0}/2}{(1+1/4)\md{0} + \md{0}/2}\cdot \frac{1}{n} \\
&=& \frac{2}{7n}
\end{eqnarray*}}
\end{proof}

\begin{lemma}
\label{ob:weak-10p}
If $\md{0} \le \dg{t}{u} < (1+1/4)\md{0}$, $w\in \nei{0}{1}{u}$ is
weakly tied to $\nei{t}{2}{u}$, and $v\in\nei{0}{2}{u}\cap
\nei{0}{1}{w}$, then the probability that $u$ connects to $v$ through
$w$ in round $t$ is at least $1/(4\md{0}^2)$.
\junk{
\[\prob{u\mbox{ connects to node in }v\mbox{ through }w\mbox{ in
    round }t} \ge \frac{1}{4\md{0}^2} \]
}
\end{lemma}
\begin{proof}
Since $w$ is weakly tied to $\nei{t}{2}{u}$ and $\dg{t}{w}$, is at
most $|\nei{t}{1}{u}| + \dgi{t}{w}{\nei{t}{2}{u}}$, we obtain that
$\dg{t}{w}$ is at most $(1+1/4)\md{0} + \md{0}/2$.  Therefore, the
probability that $u$ connects to $v$ through $w$ in round $t$ is
\begin{eqnarray*}
= \frac{1}{\dg{t}{w}^2} \ge \frac{1}{\rb{(1+1/4)\md{0} + \md{0}/2}^2} 
\;\;\; \ge \;\;\; \frac{1}{\rb{7\md{0}/4}^2} 
\;\;\; \ge \;\;\; \frac{1}{4\md{0}^2}.
\end{eqnarray*}
\junk{
\begin{eqnarray*}
& & \prob{u\mbox{ connects to node in }v\mbox{ through }w\mbox{ in
    round }t} \\
&=& \frac{1}{\dg{t}{w}^2} \\
&\ge& \frac{1}{\rb{(1+1/4)\md{0} + \md{0}/2}^2} \\
&\ge& \frac{1}{\rb{7\md{0}/4}^2} \\
&\ge& \frac{1}{4\md{0}^2} 
\end{eqnarray*}
}
\end{proof}

\begin{figure}[ht]
\begin{center}
  \includegraphics[width=6.5in]{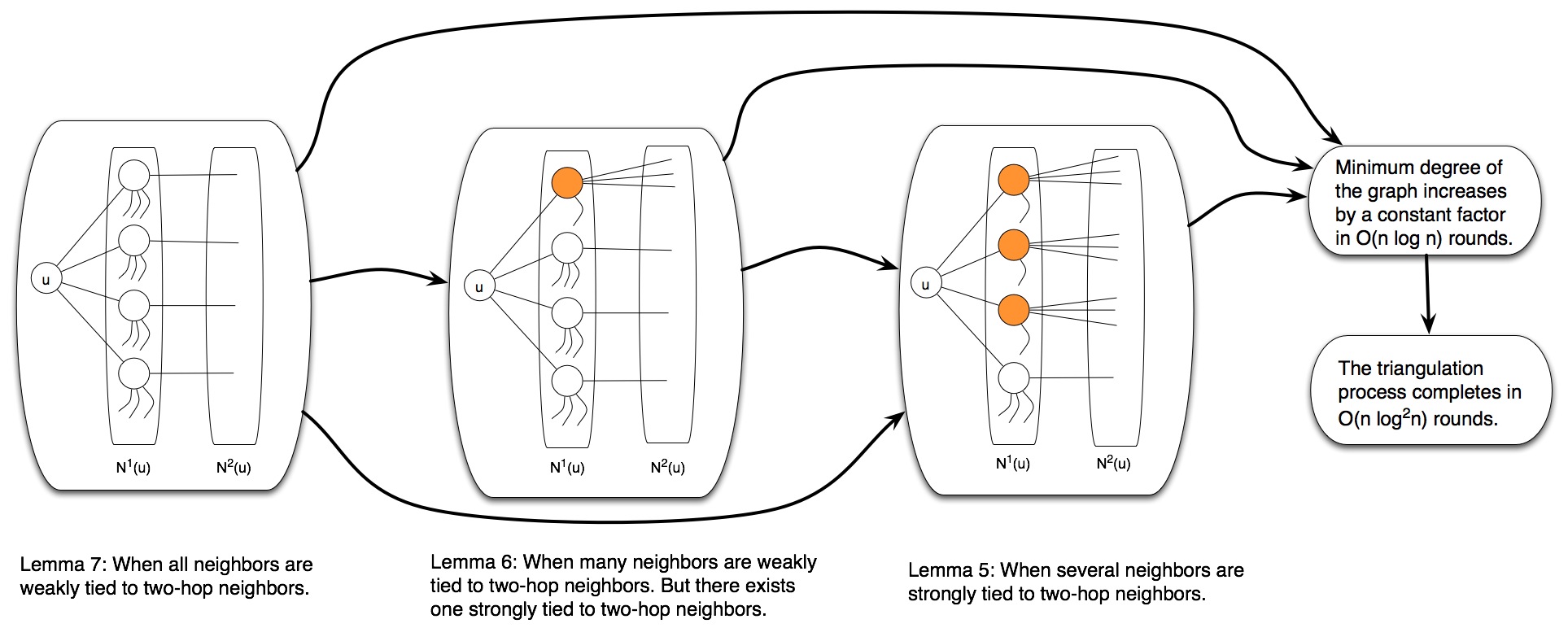}
  \caption{{\small This figure illustrates the different cases and
      relations between lemmas used in the proof of Theorem
      \ref{thm:triangulation-10p}. The shaded nodes in $\nei{t}{1}{u}$
      are strongly tied to $\nei{t}{2}{u}$. Others are weakly tied to
      $\nei{t}{2}{u}$.}}
  \label{fig:triangulation.proof-10p}
\end{center}
\end{figure}

For analyzing the growth in the degree of a node $u$, we consider
two overlapping cases.  The first case is when more than $\md{0}/4$
nodes of $\nei{t}{1}{u}$ are strongly tied to $\nei{t}{2}{u}$, and
the second is when less than $\md{0}/3$ nodes of $\nei{t}{1}{u}$
are strongly tied to $\nei{t}{2}{u}$.  The analysis for the first case
is relatively straightforward: when several neighbors of a node $u$
are strongly tied to $u$'s two-hop neighbors, then their triangulation
steps connect $u$ to a large fraction of these two-hop neighbors.

\begin{lemma}[{\bf When several neighbors are strongly tied to two-hop neighbors}]
\label{lem:triangle+case1-10p}
There exists $T = O(n \log n)$ such that if more than $\md{0}/4$ nodes
in $\nei{t}{1}{u}$ are strongly tied to $\nei{t}{2}{u}$ for all $t<
T$, then $\dg{T}{u} \ge (1+1/4)\md{0}$ with probability at least
$1-1/n^2$.
\end{lemma}
\begin{proof}
  If at any round $t< T$, $\dg{t}{u} \ge \rb{1+1/4}\md{0}$, then the
  claim of the lemma holds. In the remainder of this proof, we assume
  $\dg{t}{u} < \rb{1+1/4}\md{0}$ for all $t < T$. Let $w\in
  \nei{t}{1}{u}$ be a node that is strongly tied to
  $\nei{t}{2}{u}$. By Lemma \ref{ob:strong-10p} we know that
\[ \prob{u\mbox{ connects to a node in }\nei{t}{2}{u}\mbox{ through
  }w\mbox{ in round } t} \ge \frac{2}{7n} > \frac{1}{6n}\]
We have more than $\md{0}/4$ such $w$'s in $\nei{t}{1}{u}$, each of
which independently executes a triangulation step in any given round.
Consider a run of $T_1 = 72 n\ln n/\md{0}$ rounds.  This implies at
least $18 n \ln n$ attempts to add an edge between $u$ and a node in
$\nei{t}{2}{u}$.  Thus,
\begin{eqnarray*}
\prob{u\mbox{ connects to a node in }\nei{t}{2}{u}\mbox{ after
  }T_1\mbox{ rounds}} \ge 1-\rb{1-\frac{1}{6n}}^{18 n \ln n} \ge 1- e^{-3\ln n} = 1 - \frac{1}{n^3}.
\end{eqnarray*}
If a node that is two hops away from $u$ becomes a neighbor of $u$ by
round $t$, it is no longer in $\nei{t}{2}{u}$.  Therefore, in $T=T_1
\md{0}/4 = O(n\log n)$ rounds, $u$ will connect to at least $\md{0}/4$
new nodes with probability at least $1-1/n^2$, i.e., $\dg{T}{u}\ge
\rb{1+1/4}\md{0}$.
\end{proof}
We next consider the second case where less than $\md{0}/3$ neighbors
of a given node $u$ are strongly tied to the two-hop neighborhood of
$u$.  This case is more challenging since the neighbors of $u$ that
are weakly tied may not contribute many new edges to $u$.  We break
the analysis of this part into two subcases based on whether there is
at least one neighbor of $u$ that is strongly tied to
$\nei{0}{2}{u}$. Figure \ref{fig:triangulation.proof-10p} illustrates
the different cases and lemmas used in the proof of Theorem
\ref{thm:triangulation-10p}.

\begin{lemma}[{\bf When few neighbors are strongly tied to two-hop neighbors}]
\label{lem:triangle+extro-10p}
There exists $T=O(n\log n)$ such that if less than $\md{0}/3$ nodes
in $\nei{t}{1}{u}$ are strongly tied to $\nei{t}{2}{u}$ for all $t <
T$, and there exists a node $v_0\in \nei{0}{1}{u}$ that is strongly
tied to $\nei{0}{2}{u}$, then $\dg{T}{u} \ge \rb{1+1/8}\md{0}$ with
probability at least $1-1/n^2$.
\end{lemma}
\begin{proof}
If at any point $t < T$, $\dg{T}{u} \ge \rb{1+1/8}\md{0}$, then the
claim of the lemma holds. In the remainder of this proof, we assume
$\dg{T}{u} < \rb{1+1/8}\md{0}$ for all $t < T$.  Let $S^0_t$ denote
the set of $v_0$'s neighbors in $\nei{t}{2}{u}$ which are strongly
tied to $\nei{t}{1}{u}$ at round $t$, $W^0_t$ denote the set of $v_0$'s
neighbors in $\nei{t}{2}{u}$ which are weakly tied to $\nei{t}{1}{u}$
at round $t$.

Consider any node $v$ in $S^0_t$.  Less than $\md{0}/3$ nodes in
$\nei{t}{1}{u}$ are strongly tied to $\nei{t}{2}{u}$, thus more than
$\md{0}/2 - \md{0}/3 = \md{0}/6$ neighbors of $v$ in $\nei{t}{1}{u}$
are weakly tied to $\nei{t}{2}{u}$. Let $w$ be one such weakly tied
node. By Lemma \ref{ob:weak-10p}, the probability that $u$ connects to
$v$ through $w$ in round $t$ is at least $1/(4\md{0}^2)$.  We have at
least $\md{0}/6$ such $w$'s, each of which executes a triangulation
step each round.  Consider $T=72\md{0}\ln n$ rounds of the process.
Then the probability that $u$ connects to $v$ in $T$ rounds is at
least
\[
1-\rb{1-\frac{1}{4\md{0}^2}}^{12 \md{0}^2 \ln n} \ge 1-e^{-3\ln n} 
= 1- \frac{1}{n^3}.
\]
Thus, if $|S^0_t|\ge \md{0}/8$, in an additional $O(n\log n)$ rounds,
$\dg{T}{u}\ge (1+1/8)\md{0}$ with probability at least
$1-1/n^2$. 

Therefore, in the remainder of the proof we consider the case where
$|S^0_t| < \md{0}/8$. Define $R^0_t = R^0_{t-1} \cup W^0_t$, $R^0_0 =
W^0_0$. If at least $\md{0}/8$ nodes in $R^0_t$ are connected to
$u$ at any time, then the claim of the lemma holds. Thus, in the
following we consider the case where $|R^0_t \cap \nei{t}{1}{u}| <
\md{0}/8$. From the definition of $R^0_t$, we can derive
\[|R^0_t| \ge |W^0_t| = \dgi{t}{v_0}{\nei{t}{2}{u}} - |S^0_t| \ge
\dgi{t}{v_0}{\nei{t}{2}{u}}  - \md{0}/8 \]
At round 0, $v_0$ is strongly tied to $\nei{0}{2}{u}$,
i.e., $\dgi{0}{v_0}{\nei{0}{2}{u}} \ge \md{0}/2$. Since $\md{0}\le \dg{t}{u} <
(1+1/8)\md{0}$, we have 
\[ \dgi{t}{v_0}{\nei{t}{2}{u}} \ge \dgi{t}{v_0}{\nei{0}{2}{u}} - \md{0}/8
  \ge 3\md{0}/8\]

Let $e_1$ denote the event $\cub{u\mbox{ connects to a node in
  }R^0_t\setminus \nei{t}{1}{u}\mbox{ through }v_0\mbox{ in round }t}$.
\begin{eqnarray*}
\prob{e_1} &=& \frac{|R^0_t\setminus \nei{t}{1}{u}|}{\dg{t}{v_0}} \cdot
\frac{1}{\dg{t}{v_0}} 
\;\;\; = \;\;\; \frac{|R^0_t| - |R^0_t \cap \nei{t}{1}{u}|}{\dg{t}{v_0}} \cdot
\frac{1}{\dg{t}{v_0}} \\
& \ge & \frac{|R^0_t| - |R^0_t \cap \nei{t}{1}{u}|}{\dg{t}{v_0}} \cdot
\frac{1}{n} 
\;\;\;= \;\;\; \frac{|R^0_t| - |R^0_t \cap \nei{t}{1}{u}|}{\abs{\nei{t}{1}{u}}+\dgi{t}{v_0}{\nei{t}{2}{u}}} \cdot
\frac{1}{n} \\
& \ge & \frac{|R^0_t| - \md{0}/8}{\abs{\nei{t}{1}{u}}+\dgi{t}{v_0}{\nei{t}{2}{u}}} \cdot
\frac{1}{n} 
\;\;\; \ge \;\;\; \frac{\dgi{t}{v_0}{\nei{t}{2}{u}}  - \md{0}/8 - \md{0}/8}{\abs{\nei{t}{1}{u}}+\dgi{t}{v_0}{\nei{t}{2}{u}}} \cdot
\frac{1}{n} \\
&\ge& \frac{3\md{0}/8 - \md{0}/8 - \md{0}/8}{\abs{\nei{t}{1}{u}}+3\md{0}/8} \cdot
\frac{1}{n} 
\;\;\; \ge \;\;\;  \frac{3\md{0}/8 - \md{0}/8 - \md{0}/8}{(1+1/8)\md{0}+3\md{0}/8} \cdot
\frac{1}{n} 
\;\;\; = \;\;\;  \frac{1}{12n}
\end{eqnarray*}
Let $X_1$ be the number of rounds it takes for $e_1$ to occur. When
$e_1$ occurs, let $v_1$ denote a witness for $e_1$; i.e., if we use
$X_1$ to denote the round at which $e_1$ occurs, then let $v_1$ denote
a node in $R^0_{X_1} \setminus \nei{X_1}{1}{u}$ to which $u$ connects
through $v_0$ in round $X_1$.  Since $v_1$ is in $R^0_{X_1}$, it is
also in $W^0_{t_1}$ for some $t_1 \le X_1$; therefore, $v_1$ is
strongly tied to $\nei{t_1}{2}{u}\cap \nei{t_1}{3}{u}$. If
$\dgi{t}{v_1}{\nei{t}{2}{u}} < 3\md{0}/8$ at any point $t$, then
$\dg{t}{u}\ge (1+1/8)\md{0}$.  Thus, in the remainder of the proof, we
consider the case where $\dgi{t}{v_1}{\nei{t}{2}{u}}\ge
3\md{0}/8$. Let $S^1_t$ (resp., $W^1_t$) denote the set of $v_1$'s
neighbors in $\nei{t}{2}{u}$ that are strongly tied (resp., weakly
tied) to $\nei{t}{1}{u}$.  If $|S^1_t| \ge \md{0}/8$, then as we did
for the case $|S^0_t| \ge \md{0}/8$, we argue that in $O(n \log n)$
rounds, the degree of $u$ is at least $(1 + 1/8)\md{0}$ with
probability at least $1 - 1/n^2$.

Thus, in the remainder, we assume that $|S^1_t| < \md{0}/8$.  Define
$R^1_t = R^1_{t-1} \cup W^1_t$, $R^1_{t_1} = W^1_{t_1}$. Let $e_2$
denote the event $\cub{u\mbox{ connects to a node in }R^0_t\setminus
  \nei{t}{1}{u}(\mbox{or }R^1_t\setminus \nei{t}{1}{u})\mbox{ through
  }v_0(\mbox{or }v_1)\mbox{ in round }t}$. By the same calculation as
for $v_0$, we have $\prob{e_2} \ge 1/6n$.  Similarly, we can define
$e_3, X_3, e_4, X_4, \dots, e_{\md{0}/4}, X_{\md{0}/4}$, and obtain
that $\prob{e_i} \ge i/(12n)$. The total number of rounds for $u$ to
gain $\md{0}/4$ edges is bounded by $T=\sum_i X_i$.  By Lemma
\ref{lem:couponcorollary-10p}, $T\le 36n \ln n$ with probability at least
$1-1/n^2$, completing the proof.
\end{proof}

\begin{lemma}[{\bf When all neighbors are weakly tied to two-hop neighbors}]
\label{lem:triangle+case2-10p}
There exists $T = O(n\log n)$ such that if all nodes in
$\nei{t}{1}{u}$ are weakly tied to $\nei{t}{2}{u}$ for all $t< T$,
then $\dg{T}{u}\ge \min\cub{(1+1/8)\md{0}, n-1}$ with probability at
least $1-1/n^2$.
\end{lemma}
\begin{proof}
  If at any point $t < T$, $\dg{t}{u}\ge \min\cub{(1+1/8)\md{0},
    n-1}$, then the claim of this lemma holds.  In the remainder of
  this proof, we assume $\dg{t}{u} < \min\cub{(1+1/8)\md{0}, n-1}$ for
  all $t < T$.  \junk{ When there exists an strongly tied node in
    $\nei{t}{1}{u}$, by Lemma \ref{lem:triangle+extro-10p}, we know in
     $T=O(n\log n)$ rounds, $\dg{T}{u}\ge (1+1/8)\md{0}$, which proves
    this lemma. Thus, we now consider the case where all nodes in
    $\nei{0}{1}{u}$ are weakly tied to $\nei{t}{2}{u}$ for $t < T$.  }
  In the following, we first show, any node $v\in \nei{0}{2}{u}$
  will have at least $\md{0}/4$ edges to $\nei{T_1}{1}{u}$, where $T_1
  = O(n\log n)$. After that, $v$ will connect to $u$ in $T_2 = O(n\log
  n)$ rounds. Therefore, the total number of rounds used for $v$ to
  connect to $u$ is $T_3 = T_1 + T_2 = O(n\log n)$.

Node $v$ at least connects to one node in $\nei{0}{1}{u}$. Call it
$w_1$. Because all nodes in $\nei{t}{1}{u}$ are weakly tied to
$\nei{t}{2}{u}$, we have $\dgi{t}{w_1}{\nei{t}{1}{u}} \ge \md{0} -
\md{0}/2 = \md{0}/2$. If $\dgi{t}{w_1}{\nei{t}{1}{u}\setminus
  \nei{t}{1}{v}} < \md{0}/4$, then $v$ already has $\md{0}/4$ edges to
$\nei{t}{1}{u}$. Thus, in the following we consider the case where
$\dgi{t}{w_1}{\nei{t}{1}{u}\setminus \nei{t}{1}{v}} \ge \md{0}/4$. Let
$e_1$ denote the event $\cub{v\mbox{ connects to a node in
  }\nei{t}{1}{u}\setminus \nei{t}{1}{v}\mbox{ through }w_1}$.
\begin{eqnarray*}
\prob{e_1} &=& \frac{\dgi{t}{w_1}{\nei{t}{1}{u}\setminus\nei{t}{1}{v}}}{\dg{t}{w_1}} \cdot
\frac{1}{\dg{t}{w_1}} 
\;\;\; \ge \;\;\; \frac{\dgi{t}{w_1}{\nei{t}{1}{u}\setminus \nei{t}{1}{v}}}{\abs{\nei{t}{1}{u}} +
  \dgi{t}{w_1}{\nei{t}{2}{u}}} \cdot \frac{1}{\dg{t}{w_1}} \\
&\ge& \frac{\md{0}/4}{(1+1/8)\md{0} + \md{0}/2} \cdot
\frac{1}{\dg{t}{w_1}} 
\;\;\; \ge \;\;\; \frac{2}{13} \cdot \frac{1}{n} 
\;\;\; > \;\;\; \frac{1}{7n}
\end{eqnarray*}
Let $X_1$ be the number of rounds needed for $e_1$ to occur. When
$e_1$ occurs, let $w_2$ denote a witness for $e_1$; i.e., let $w_2$
denote a vertex in $\nei{t}{1}{u} \setminus \nei{t}{1}{v}$ to which
$v$ connects.  Note that here the value of $t$ is the round at which
the event occurs.  By our choice, $w_2$ is also weakly tied to
$\nei{t}{2}{u}$.  By an argument similar to the one in the above
paragraph, we have $\dgi{t}{w_2}{\nei{t}{1}{u}\setminus \nei{t}{1}{v}}
\ge \md{0}/4$. Let $e_2$ denote the event $\cub{v\mbox{ connects to a
    node in }\nei{t}{1}{u}\mbox{ through }w_1\mbox{ or }w_2}$. We have
$\prob{e_2} \ge 2/(7n)$.  Let $X_2$ be the number of rounds needed for
$e_2$ to occur. Similarly, we can define $e_3, X_3, \dots,
e_{\md{0}/4}, X_{\md{0}/4}$ and show $\prob{e_i} \ge i/(7n)$. Set $T_1
= \sum_i X_i$, which is the bound on the number of rounds needed for
$v$ to have at least $\md{0}/4$ neighbors in $\nei{t}{1}{u}$. By Lemma
\ref{lem:couponcorollary-10p}, $T_2\le 28n\ln n$ with probability at
least $1-1/n^3$. Now we show $v$ will connect to $u$ in $T_2$ rounds
after this. Notice that, all $w_i$'s are still weakly tied to
$\nei{t}{2}{u}$. By Lemma \ref{ob:weak-10p}, the probability that $u$
connects to $v$ through $w_i$ in round $t$ is at least
$1/(4\md{0}^2)$.  We have $w_1,w_2,\dots,w_{\md{0}/4}$ independently
executing a triangulation step each round. Consider $T_2 = 48\md{0}\ln
n$ rounds of the process. Then,
\begin{eqnarray*}
\prob{u\mbox{ connects to }v\mbox{ in }T_2\mbox{ rounds}} 
\ge 1-\rb{1-\frac{1}{4\md{0}^2}}^{12\md{0}^2\ln n} 
\ge 1-\frac{1}{n^3}.
\end{eqnarray*}
We have shown for any node $v\in \nei{0}{2}{u}$, it will connect to
$u$ in round $T_3=T_1+T_2$ with probability at least $1-1/n^3$. This
implies in round $T_3$, $u$ will connect to all nodes in
$\nei{0}{2}{u}$ with probability at least
$1-\abs{\nei{0}{2}{u}}/n^3$. Then, $\nei{0}{2}{u} \subseteq
\nei{T_3}{1}{u}, \nei{0}{3}{u} \subseteq \nei{T_3}{1}{u}\cup
\nei{T_3}{2}{u}, \nei{0}{4}{u} \subseteq \nei{T_3}{1}{u}\cup
\nei{T_3}{2}{u}\cup \nei{T_3}{3}{u}$.  We apply the above analysis
twice, and obtain that in round $T=3T_3=O(n\log n)$,
$\nei{0}{2}{u}\cup \nei{0}{3}{u}\cup \nei{0}{4}{u} \subseteq
\nei{T}{1}{u}$ with probability at least $1-\abs{\nei{0}{2}{u}\cup
  \nei{0}{3}{u}\cup \nei{0}{4}{u}}/n^3\ge 1-1/n^2$. By Lemma
\ref{lem:moreneighbor-10p}, $\abs{\nei{0}{2}{u}\cup \nei{0}{3}{u}\cup
  \nei{0}{4}{u}} \ge\min\cub{2\md{0}, n-1}$, thus completing the
proof.
\end{proof}

\begin{theorem}[{\bf Upper bound for triangulation process}]
\label{thm:triangulation-10p}
  For any connected undirected graph, the triangulation process
  converges to a complete graph in $O(n\log^2 n)$ rounds with high
  probability.
\end{theorem}
\begin{proof}
  We first show that in $O(n\log n)$ rounds, either the graph becomes
  complete or its minimum degree increases by a factor of at least
  $1/12$. Then we apply this argument $O(\log n)$ times to complete
  the proof.

  For each $u$ where $\dg{0}{u} <\min\cub{(1+1/8)\md{0}, n-1}$, we
  consider the following 2 cases.  The first case is if more than
  $\md{0}/3$ nodes in $\nei{0}{1}{u}$ are strongly tied to
  $\nei{0}{2}{u}$.  By Lemma \ref{lem:triangle+case1-10p}, there exists
  $T=O(n\log n)$ such that if at least $\md{0}/4$ nodes in
  $\nei{t}{1}{u}$ are strongly tied to $\nei{t}{2}{u}$ for $t<T$, then
  $\dg{T}{u}\ge (1+1/8)\md{0}$ with probability at least $1-1/n^2$.
  Whenever the condition is not satisfied, i.e., less than $\md{0}/4$
  nodes in $\nei{t}{1}{u}$ are strongly tied to $\nei{t}{2}{u}$, it
  means more than $\md{0}/3-\md{0}/4=\md{0}/12$ strongly tied nodes
  became weakly tied. By the definitions of strongly tied and weakly
  tied, this implies $\dg{T}{u} \ge (1+1/12)\md{0}$.

The second case is if less than $\md{0}/3$ nodes in $\nei{0}{1}{u}$
are strongly tied to $\nei{0}{2}{u}$.  By
Lemmas~\ref{lem:triangle+extro-10p} and~\ref{lem:triangle+case2-10p}, we know
that there exists $T=O(n\log n)$ such that if we remain in this case
for $T$ rounds, then $\dg{T}{u}\ge \min\cub{ (1+1/8)\md{0}, n-1}$ with
probability at least $1-1/n^2$.  Whenever the condition is not
satisfied, i.e., more than $\md{0}/3$ nodes in $\nei{t}{1}{u}$ are
strongly tied to $\nei{t}{2}{u}$, we move to the analysis in the first
case, and $\dg{T}{u}\ge (1+1/8)\md{0}$ in $T=O(n\log n)$ rounds with
probability at least $1-1/n^2$.

Combining the above 2 cases and applying a union bound, we obtain
$\md{T} \ge \min\cub{(1+1/8)\md{0},n-1}$ in $T=O(n\log n)$ rounds with
probability at least $1-1/n$.  We now apply the above argument $O(\log
n)$ times to obtain the desired $O(n\log^2 n)$ upper bound.
\end{proof}

\junk{
\begin{lemma}
[{\bf Lower bound for triangulation process}]
 For any connected undirected graph $G$ with minimum degree $\alpha
  n$ for some constant $\alpha$, with high probability the
  triangulation process takes $\Omega(n\log n)$ to complete when the
  number of missing edges (compared with complete graph) is at least
  $n$.
\end{lemma}
\begin{proof}
  For any missing edge $(u,v)$, look at a common neighbor of $u$ and
  $v$, referred as $w$. Then node $w$ connecting $u$ and $v$ in a
  triangulation step is a Bernoulli experiment. Because the minimum
  degree of $G$ is $\alpha n$, the probability of a success outcome is
  at most $1/ \rb{\alpha n}^2$. Let $k$ be the number of missing
  edges. In the following we argue there has to be $\Omega(n^2\ln n)$
  such Bernoulli trials before all $k$ missing edges to be added. The
  triangulation process can only carry out at most $n$ such Bernoulli
  trials in each step, thus we complete the proof of $\Omega(n\ln n)$
  lower bound. 

  Let random variable $X_i$ denote the number of trials for the $i$th
  missing edge to be added. Let $T$ be $\beta n^2\ln n$, where $\beta$
  is a constant to be specified later.
\begin{eqnarray*}
\prob{X_1\le T, X_2\le T, \dots, X_k \le T} &=& \sqb{1-\rb{1-\frac{1}{\rb{\alpha n}^2}}^T}^k \\
&=& \sqb{1-\rb{1-\frac{1}{\rb{\alpha n}^2}}^{\beta n^2\ln n}}^k \\
&\rightarrow& \rb{1-e^{-\frac{\beta}{\alpha^2} \ln n}}^k \\
&=& \rb{1-\frac{1}{n^{\beta / \alpha^2}}}^k \\
&\le& \rb{1-\frac{1}{n^{\beta / \alpha^2}}}^n
\end{eqnarray*}
Set $\beta = \alpha^2/2$. Then,
\[\prob{X_1\le T, X_2\le T, \dots, X_k \le T} \le
\rb{1-\frac{1}{\sqrt{n}}}^n \le e^{-\sqrt{n}} \rightarrow 0\]
This completes the proof of this lemma.
\end{proof}
}

\subsection{Lower bound}
The proof of the following theorem is in
Appendix~\ref{app:triangulation}.

\begin{theorem}
[{\bf Lower bound for triangulation process}]
\label{thm:tri+lower-10p}
For any connected undirected graph $G$ that has $k \ge 1$ edges less
than the complete graph the triangulation process takes $\Omega(n \log
k)$ steps to complete with probability at least $1 - O\rb{e^{-k^{1/4}}}$.
\end{theorem}

\junk{
\begin{proof}
  We first observe that during the triangulation process there is a 
  time $t$ when the number of missing edges is at least $m = O(\sqrt{k})$
  and the minimum degree is at least $n/3$. If $k<\frac{2}{3} n$ then 
  this is true initially and for larger $k$ this is true at the first 
  time $t$ the minimum degree is large enough. The second case follows since
  the degree of a node (and thus also the minimum degree) can at most
  double in each step guaranteeing that the minimum degree is not larger
  than $\frac{2}{3}n$ at time $t$ also implying that at least $\frac{n}{3} = \Omega(\sqrt{k})$
  edges are still missing. 
  
  Given the bound on the minimum degree any missing edge $\{u,v\}$ is added 
  by a fixed node $w$ with probability at most $\frac{9}{2n^2}$.
  Since there are at most $n-2$ such nodes the probability that a missing edge
  gets added is at most $\frac{9}{2n}$. To analyze the time needed for all missing
  edges to be added we denote with $X_i$ the random variable
  counting the number of steps needed until the $i$th of the $m$ missing edges is added.
  We would like to analyze $\prob{X_1\le T, X_2\le T, \dots, X_m \le T}$ for an 
  appropriately chosen number of steps $T$. Note that the
  events $X_i < T$ and $X_j <T$ are not independent and indeed can be positively
  or negatively correlated. Nevertheless, independent of the conditioning onto 
  any of the events $X_j < T$, we have that $\prob{X_1 \le T} \leq 1 - (1 - \frac{9}{2n})^T \leq 1 - \frac{1}{\sqrt{m}}$ for an appropriately chosen $T = \Omega(n \log m)$, where $m$ is again the number of missing edges at time $t$. Thus,
$$\prob{X_1\le T, X_2\le T, \dots, X_m \le T} =$$
$$= \prob{X_1\le T| X_2\le T, \dots, X_m \le T} \cdot \prob{X_2 \le T | X_3, \dots, X_m \le T} \cdot \ldots \cdot \prob{X_m \le T} $$
$$\leq \rb{1-\frac{1}{\sqrt{m}}}^m \; = \; O\rb{e^{-\sqrt{m}}} \; = \; O\rb{e^{-k^{1/4}}}$$

This shows that the triangulation process takes with probability at least $1-O\rb{e^{-k^{1/4}}}$ at least $\Omega(n \log m) = O(n \log k)$ steps to complete. 
\end{proof}
}

\section{The two-hop walk: Discovery through pull}
\label{sec:2hop-10p}
In this section, we analyze the two-hop walk process on undirected
connected graphs, which is described by the following simple
iteration: In each round, for each node $u$, we add edge $(u,w)$
where $w$ is drawn uniformly at random from $\nei{t}{1}{v}$, where $v$
is drawn uniformly at random from $\nei{t}{1}{u}$.  The two-hop walk
yields the following pull-based resource discovery protocol.  In each
round, each node $u$ contacts a random neighbor $v$, receives the
identity of a random neighbor $w$ of $v$, and sends its identity to
$w$.  The main result of this section is that the two-hop walk process
transforms an arbitrary connected $n$-node graph to a complete graph
in $O(n \log^2 n)$ rounds with high probability.  We also establish an
$\Omega(n\log n)$ lower bound on the two-hop walk for almost all
$n$-node graphs.

As for the triangulation process, we establish the $O(n \log^2 n)$
upper bound by showing that the minimum degree of the graph increases
by a constant factor (or equals $n-1$) in $O(n \log n)$ rounds with
high probability.  For analyzing the growth in the degree of a node
$u$, we consider two overlapping cases.  The first case is when the
two-hop neighborhood of $u$ is not too large, i.e., $|\nei{t}{2}{u}| <
\md{0}/2$, and the second is when the two-hop neighborhood of $u$ is
not too small, i.e., $|\nei{t}{2}{u}| \ge \md{0}/4$.  The proofs of
the following three claims that establish the upper bound are deferred
to Appendix~\ref{app:2hop}.

\begin{lemma}[{\bf When the two-hop neighborhood is not too large}]
\label{lem:mingrowth-1-10p}
There exists $T=O(n\log n)$ such that either $\abs{\nei{T}{2}{u}} \ge
\md{0}/2$ or $\dg{T}{u} \ge \min\cub{2\md{0},n-1}$ with probability at
least $1-1/n^2$.
\end{lemma}

\begin{lemma}[{\bf When the two-hop neighborhood is not too small}]
\label{lem:mingrowth-2-10p}
There exists $T=O(n\log n)$ such that either $\abs{\nei{T}{2}{u}} <
\md{0}/4$ or $\dg{T}{u} \ge \min\cub{(1+1/8)\md{0},n-1}$, with probability at least $1-1/n^2$.
\end{lemma}

\begin{theorem}[{\bf Upper bound for two-hop walk process}]
\label{thm:graph+randwalk-10p}
For connected undirected graphs, the two-hop walk process completes in
$O(n\log^2 n)$ rounds with high probability.
\end{theorem}

The proof of Theorem~\ref{thm:hop+lower-10p} is essentially the same as
that for Theorem \ref{thm:tri+lower-10p}, and is omitted.

\begin{theorem}[{\bf Lower bound for two-hop walk process}]
\label{thm:hop+lower-10p}
For any connected undirected graph $G$ that has $k \ge 1$ edges less
than the complete graph the two-hop process takes $\Omega(n \log
k)$ steps to complete with probability at least $1 - O\rb{e^{-k^{1/4}}}$.
\end{theorem}

\section{Two-hop walk in directed graphs}
\label{sec:directed-10p}
In this section, we analyze the two-hop walk process in directed
graphs.  We say that the process terminates at time $t$ if for every
node $u$ and every node $v$, $\gf{t}$ contains the edge $(u,v)$
whenever $u$ has a path to $v$ in $\gf{0}$.

\begin{theorem}
\label{thm:directed.upper-10p}
On any $n$-node directed graph, the two-hop walk terminates in
$O(n^2 \log n)$ rounds with high probability.  Furthermore, there
exists a (weakly connected) directed graph for which the process takes
$\Omega(n^2 \log n)$ rounds to terminate.
\end{theorem}

The lower bound in the above theorem, whose proof is deferred to
Appendix~\ref{app:directed}, takes advantage of the fact that the
initial graph is not strongly connected.  Extending the above analysis
for strongly connected graphs appears to be much more difficult since
the events corresponding to the addition of new edges interact in
significant ways.  We present an $\Omega(n^2)$ lower bound for a
strongly connected graph by a careful analysis that tracks the event
probabilities with time and takes dependencies into account.  The
graph $\gf{0}$, depicted in Figure~\ref{fig:digraph+lower-10p}, is
similar to the example in~\cite{leighton} used to establish an
$\Omega(n)$ lower bound on the Random Pointer Jump algorithm, in which
each node gets to know all the neighbors of a random neighbor in each
step.  Since the graphs are constantly changing over time in both the
processes, the dynamic edge distributions differ significantly in the
two cases, and we need a substantially different analysis.  Due to
space constraints, we defer the proof to Appendix~\ref{app:directed}.
\begin{theorem}
\label{thm:directed.lower-10p}
There exists a strongly connected directed graph for which
the expected number of rounds taken by the two-hop process is
$\Omega(n^{2})$.
\end{theorem}

\begin{figure}[ht]
\begin{center}
  \includegraphics[width=5in]{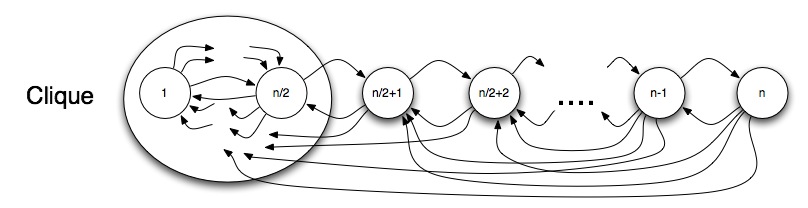}
  \caption{Lower bound example for two-hop walk process in directed graphs}
  \label{fig:digraph+lower-10p}
\end{center}
\end{figure}

\junk{
\begin{proofsketch}
The graph $\gf{0}=(V,E)$ is depicted in Figure~\ref{fig:digraph+lower-10p}
and formally defined as $\gf{0} = (V,E)$ where $V=\cub{1,2,\dots,n}$
with $n$ being even, and
\[
E = \cub{\rb{i,j}: 1 \le i,j \le n/2} \cup \cub{\rb{i,i+1} : n/2 \le i <
  n} \cup \cub{\rb{i,j}: i > j, i > n/2, i,j\in V}.
\]
We first establish an upper bound on the probability that edge
$(i,i+h)$ is added by the start of round $t$, for given $i$, $1 \le i
\le n -h$.  Let $\prb{h}{t}$ denote this probability.  The following
base cases are immediate: $\prb{h}{0}$ is $1$ for $h = 1$ and $h < 0$,
and $0$ otherwise.  Next, the edge $(i,i+h)$ is in $\gf{t+1}$ if and
only if $(i,i+h)$ is either in $\gf{t-1}$ or added in round $t$.  In
the latter case, $(i,i+h)$ is added by a two-hop walk $i \rightarrow i
+ k \rightarrow i + h$, where $-i < k \le n - i$.  Since the
out-degree of every node is at least $n/2$, for any $k$ the
probability that $i$ takes such a walk is at most $4/n^2$.
\begin{eqnarray}
\prb{h}{t+1} & \le & \prb{h}{t} + \frac{4}{n^2} \sum_{k > -i}^{n-i} \prb{k}{t}\prb{h-k}{t} = \prb{h}{t} + \frac{4}{n^2} \left(\sum_{k = 1}^{i-1} \prb{h+k}{t} + \sum_{k=1}^{h-1} \prb{k}{t}\prb{h-k}{t} + \sum_{k=h+1}^{n-i} \prb{k}{t}\right) \label{eqn:lower-10p}
\end{eqnarray}
In Appendix~\ref{app:directed}, we show by induction on $t$ that
\begin{eqnarray}
\prb{h}{t} \le \left(\frac{\alpha t}{n^2}\right)^{h-1}, \mbox{ for all } t \le \eps n^2  \label{eqn:prb-10p}
\end{eqnarray}
where $\alpha$ and $\eps$ are positive constants that are specified
later.
\junk{
The induction base is immediate.  For the induction step, we use the
induction hypothesis for $t$ and Equation~\ref{eqn:lower-10p} and bound
$\prb{h}{t+1}$ as follows.
\begin{eqnarray*}
\prb{h}{t+1} & \le & \left(\frac{\alpha t}{n^2}\right)^{h-1} + \frac{4}{n^2} \left(\sum_{k = 1}^{i-1} \left(\frac{\alpha t}{n^2}\right)^{h+k-1} + \sum_{k=1}^{h-1} \left(\frac{\alpha t}{n^2}\right)^{k-1} \left(\frac{\alpha t}{n^2}\right)^{h-k-1} + \sum_{k=h+1}^{n-i} \left(\frac{\alpha t}{n^2}\right)^{k-1}\right)\\
& \le & \left(\frac{\alpha t}{n^2}\right)^{h-1} + \frac{4}{n^2} \left( (h-1)\left(\frac{\alpha t}{n^2}\right)^{h-2} + \left(\frac{\alpha t}{n^2}\right)^{h} \frac{2}{1 - \alpha t/n^2}\right)\\
& \le & \left(\frac{\alpha t}{n^2}\right)^{h-1} + (h-1)\left(\frac{\alpha t}{n^2}\right)^{h-2}\frac{1}{n^2} \left(4 + \frac{4\eps^2}{(1 - \alpha\eps)}\right)\\
& \le & \left(\frac{\alpha t}{n^2}\right)^{h-1} + (h-1) \left(\frac{\alpha t}{n^2}\right)^{h-2}\frac{\alpha}{n^2}\\
& \le & \left(\frac{\alpha (t+1)}{n^2}\right)^{h-1}.
\end{eqnarray*}
(In the second inequality, we combine the first and third summations
and bound them by their infinite sums.  In the third inequality, we
use $t \le \eps n^2$.  For the fourth inequality, we set $\alpha$
sufficiently large so that $\alpha \ge 4 + 4/(1-\alpha \eps)$.  The
final inequality follows from Taylor series expansion.)
}

For an integer $x$, let $C_x$ denote the cut $(\{u: u \le x\}, \{v, v
> x\})$.  We say that a cut $C_x$ is {\em untouched}\/ at the start of
round $t$ if the only edge in $\gf{t}$ crossing the cut $C_x$ is the
edge $(x, x+1)$; otherwise, we say $C_x$ is {\em touched}.  Let $X$
denote the smallest integer such that $C_X$ is untouched.  We note
that $X$ is a random variable that also varies with time.  Initially,
$X = n/2$.

We divide the analysis into several phases, numbered from $0$.  A
phase ends when $X$ changes.  Let $X_i$ denote the value of $X$ at the
start of phase $i$; thus $X_0 = n/2$.  Let $T_i$ denote the number of
rounds in phase $i$.  A new edge is added to the cut $C_{X_i}$ only if
either $X_i$ selects edge $(X_i,X_i+1)$ as its first hop or a node $u
< X_i$ selects $u \rightarrow X_i \rightarrow X_i + 1$.  Since the
degree of every node is at least $n/2$, the probability that a new
edge is added to the cut $C_i$ is at most $2/n + n(4/n^2) = 6/n$,
implying that $E[T_i] \ge n/6$.

We now place a bound on $X_{i+1}$.  Fix a round $t \le \eps n^2$, and
let $E_x$ denote the event that $C_x$ is touched by round $t$.  We
first place an upper bound on the probability of $E_x$ for arbitrary
$x$ using Equation~\ref{eqn:prb-10p}.
\[
\Pr[E_x] \le \sum_{h \ge 2} h \left(\frac{\alpha t}{n^2} \right)^{h-1} \le \frac{\alpha t (4 - 3(\alpha t)/n^2 + (\alpha t)^2/n^4)}{n^2(1-(\alpha t)/n^2)^3},
\]
for $t \le \eps n^2$, where we use the inequality $\sum_{h \ge 2} h^2
\delta^h = \delta(4-3\delta+\delta^2)/(1-\delta)^3$ for $0 < \delta <
1$.  We set $\eps$ sufficiently small so that
$(4-3\eps+\eps^2)/(1-\eps)^3 \le 5$, implying that the above
probability is at most $5\eps$.

If $E_x$ were independent from $E_y$ for $x \neq y$, then we can
invoke a straightforward analysis using a geometric probability
distribution to argue that $E[X_{i+1} - X_i]$ is at most $1/(1-5\eps)
= O(1)$; to see this, observe that $X_{i+1} - X_i$ is stochastically
dominated by the number of tosses of a biased coin needed to get one
head, where the probability of tail is $5\eps$.  The preceding
independence does not hold, however; in fact, for $y > x$, $\Pr[E_y
  \mod E_x] > \Pr[E_y]$.  We show that the impact of this correlation
is very small when $x$ and $y$ are sufficiently far apart.  We
consider a sequence of cuts $C_{x_1}, C_{x_2}, \ldots, C_{x_\ell},
\ldots$ where $x_\ell = x_{\ell-1} + c\ell$, for a constant $c$ chosen
sufficiently large, and we set $x_0 = X_i + 2$.  In
Appendix~\ref{app:directed}, we bound the conditional probability of
$E_{x_\ell}$ given $E_{x_{\ell -1}} \cap E_{x_{\ell-2}} \cdots \cap
E_{x_1}$ to be at most $6\eps$, by setting $\eps$ and $c$
appropriately.  Since $X_{i+1}$ is at most the smallest $x_\ell$ such
that $C_{x_\ell}$ is untouched, we obtain that
\begin{eqnarray*}
E[X_{i+1} - X_i] \le 2 + \sum_{\ell \ge 2} (6 \eps)^\ell c\ell^2 \le c',
\end{eqnarray*}
for a constant $c'$ chosen sufficiently large.  We thus obtain that
after $\eps' n$ phases, $E[X]$ is at most $\eps'c' n$, where $\eps'$
is chosen sufficiently small so that $n - E[X]$ is $\Omega(n)$.  Since
the expected length of each phase is at least $n/6$, it follows that
the expected number of rounds it takes for the two-hop process to
complete is $\Omega(n^2)$.
\end{proofsketch}
}

\section{Conclusion}
We have analyzed two natural gossip-based discovery processes in
networks and showed almost-tight bounds on their convergence in
arbitrary networks.  Our processes are motivated by the resource
discovery problem in distributed networks as well as by the evolution
of social networks.  We would like to study variants of the processes
that take into account failures associated with forming connections,
the joining and leaving of nodes, or having only a subset of
nodes to participate in forming connections.  We believe our
techniques can be extended to analyze such situations as well.  From a
technical standpoint, the main problem left open by our work is to
resolve the logarithmic factor gap between the upper and lower bounds.
It is not hard to show that from the perspective of increasing the
minimum degree by a constant factor, our analysis is tight up to
constant factors.  It is conceivable, however, that a sharper upper
bound can be obtained by an alternative analysis that uses a
``smoother'' measure of progress.

\appendix

\section{Proofs for Section~\ref{sec:prelim}}
\label{app:prelim}
\begin{LabeledProof}{Lemma~\ref{lem:moreneighbor-10p}}
  If $\nei{t}{3}{u}$ is not an empty set, consider node $v\in
  \nei{t}{3}{u}$. Since $\dg{t}{v} \ge \md{t}$, we have
  $\abs{\cup_{i=2}^4 \nei{t}{i}{u}} \ge
  \md{t}$. $\abs{\nei{t}{1}{u}}\ge \md{t}$ because $\dg{t}{u}\ge
  \md{t}$. We also know $\nei{t}{1}{u}$ and $\cup_{i=2}^4
  \nei{t}{i}{u}$ are disjoint. Thus, $\abs{\cup_{i=1}^4 \nei{t}{i}{u}}
  \ge 2\md{t}$.  If $\nei{t}{3}{u}$ is an empty set, then
  $\nei{t}{1}{u}\cup \nei{t}{2}{u}=n-1$ because $\gf{t}$ is
  connected. Thus $\abs{\cup_{i=1}^4 \nei{t}{i}{u}} = n-1$. Combine
  the above 2 cases, we complete the proof of this lemma.
\end{LabeledProof}

\begin{LabeledProof}{Lemma~\ref{lem:moreneighbor-10p}}
  Since $X$ only increases with $k$, with out loss of generality
  assume that $k = m$. Now we can view this as {\em coupon collector
    problem} \cite{upfal} where $X_{m+1-i}$ is the number of steps to
  collect the $i$th coupon. Consider the probability of not obtaining
  the $i$th coupon after $(c+1)n\ln n$ steps. This probability is
\[\rb{1-\frac{1}{n}}^{(c+1)n\ln n} < e^{-(c+1)\ln n} =
\frac{1}{n^{c+1}}\] By union bound, the probability that some coupon
has not been collected after $(c+1)n\ln n$ steps is less than
$1/n^c$. And this completes the proof of this lemma.
\end{LabeledProof}

\section{Lower bound proof for the triangulation process}
\label{app:triangulation}
\begin{LabeledProof}{Theorem~\ref{thm:tri+lower-10p}}
  We first observe that during the triangulation process there is a 
  time $t$ when the number of missing edges is at least $m = O(\sqrt{k})$
  and the minimum degree is at least $n/3$. If $k<\frac{2}{3} n$ then 
  this is true initially and for larger $k$ this is true at the first 
  time $t$ the minimum degree is large enough. The second case follows since
  the degree of a node (and thus also the minimum degree) can at most
  double in each step guaranteeing that the minimum degree is not larger
  than $\frac{2}{3}n$ at time $t$ also implying that at least $\frac{n}{3} = \Omega(\sqrt{k})$
  edges are still missing. 
  
  Given the bound on the minimum degree any missing edge $\{u,v\}$ is added 
  by a fixed node $w$ with probability at most $\frac{9}{2n^2}$.
  Since there are at most $n-2$ such nodes the probability that a missing edge
  gets added is at most $\frac{9}{2n}$. To analyze the time needed for all missing
  edges to be added we denote with $X_i$ the random variable
  counting the number of steps needed until the $i$th of the $m$ missing edges is added.
  We would like to analyze $\prob{X_1\le T, X_2\le T, \dots, X_m \le T}$ for an 
  appropriately chosen number of steps $T$. Note that the
  events $X_i < T$ and $X_j <T$ are not independent and indeed can be positively
  or negatively correlated. Nevertheless, independent of the conditioning onto 
  any of the events $X_j < T$, we have that $\prob{X_1 \le T} \leq 1 - (1 - \frac{9}{2n})^T \leq 1 - \frac{1}{\sqrt{m}}$ for an appropriately chosen $T = \Omega(n \log m)$, where $m$ is again the number of missing edges at time $t$. Thus,
$$\prob{X_1\le T, X_2\le T, \dots, X_m \le T} =$$
$$= \prob{X_1\le T| X_2\le T, \dots, X_m \le T} \cdot \prob{X_2 \le T | X_3, \dots, X_m \le T} \cdot \ldots \cdot \prob{X_m \le T} $$
$$\leq \rb{1-\frac{1}{\sqrt{m}}}^m \; = \; O\rb{e^{-\sqrt{m}}} \; = \; O\rb{e^{-k^{1/4}}}$$

This shows that the triangulation process takes with probability at
least $1-O\rb{e^{-k^{1/4}}}$ at least $\Omega(n \log m) = O(n \log k)$
steps to complete.
\end{LabeledProof}

\section{Proofs for the two-hop walk}
\label{app:2hop}
\begin{LabeledProof}{Lemma~\ref{lem:mingrowth-1-10p}}
By the definition of $\md{0}$, $\dg{0}{w} \ge \md{0}$ for all
$w$ in $\nei{0}{1}{u}$.  Let $X$ be the first round at which
$|\nei{X}{2}{u}| \ge \md{0}/2$.  We consider two cases.  If $X$ is at
most $c n \log n$ for a constant $c$ to be specified later, then the
claim of the lemma holds.  In the remainder of this proof we consider
the case where $X$ is greater than $cn \log n$; thus, for $0 \le t \le
cn\log n$, $|\nei{t}{2}{u}| <\md{0}/2$.

Consider any node $w$ in $\nei{0}{1}{u}$.  Since $\dg{0}{w} \ge
\md{0}$ and $|\nei{t}{2}{u}| <\md{0}/2$, $w$ has at least $\md{0}/2$
edges to nodes in $\nei{0}{1}{u}$.  Fix a node $v$ in
$\nei{0}{2}{u}$.  In the following, we first show that in $O(n\log n)$
rounds, $v$ is strongly tied to the neighbors of $u$ with probability
at least $1-1/n^3$.  Let $T_1$ denote the first round at which $v$ has
is strongly tied to $\nei{T_1}{1}{u}$, i.e., when $|\nei{T_1}{1}{v}
\cap \nei{T_1}{1}{u}| \ge \md{0}/4$.  We know that $v$ has at least
one neighbor, say $w_1$, in $\nei{0}{1}{u}$.  Consider any $t < T_1$.
Since $v$ is weakly tied to $\nei{0}{1}{u}$ at time $t$, $w_1$ has at
least $\md{0}/4$ neighbors in $\nei{0}{1}{u}$ which do not have an
edge to $v$ at time $t$. This implies
\[\prob{v\mbox{ connects to a node in }\nei{0}{1}{u}\mbox{ through
  }w_1\mbox{ in round $t$}} \ge \frac{1}{n}\cdot\frac{1}{4} = \frac{1}{4n}\]

Let $e_1$ denote the event $\cub{v\mbox{ connects to a node in
  }\nei{0}{1}{u}}$, and $X_1$ be the number of rounds for $e_1$ to
occur.  When $e_1$ occurs, let $w_2$ denote a witness for $e_1$; i.e.,
$w_2$ is a node in $\nei{0}{1}{u}$ to which $v$ connects in round
$X_1$.  We note that $w_1,w_2\in \nei{0}{1}{u}\subseteq
\nei{X_1}{1}{u}$. If $v$ is weakly tied to $\nei{X_1}{1}{u}$, both
$w_1$ and $w_2$ have at least $\md{0}/4$ neighbors in
$\nei{X_1}{1}{u}$ that do not have an edge to $v$ yet. Let $e_2$
denote the event $\cub{v\mbox{ connects to a node in
  }\nei{X_1}{1}{u}}$, and $X_2$ be the number of rounds for $e_2$ to
occur. Then $\prob{e_2} = 2\prob{e_1} \ge 1/2n$.  Similarly, we define
$e_3,X_3,\dots, e_{\md{0}/4},X_{\md{0}/4}$ and obtain $\prob{e_i} \ge
i/(4n)$.  We now apply Lemma~\ref{lem:couponcorollary-10p} to obtain that
$X_1+X_2+\ldots X_{\md{0}/4}$ is at most $16n\ln n$ with probability
at least $1 - 1/n^3$.  Thus, with probability at least $1 -
|\nei{0}{2}{u}|/n^3$, $T_1 \le 16n \ln n$.  After $T_1$ rounds, we
obtain that for any $v \in \nei{0}{2}{u}$,
\[\prob{u\mbox{ connects to }v\mbox{ in a single round}} \ge \frac{\md{0}/4}{2\md{0}}\cdot\frac{1}{n} = \frac{1}{8n}.\]
which implies that with probability at least $1 - 1/n^3$, $u$ has an
edge to every node in $\nei{0}{2}{u}$ in another $T_2 \le 24 n \ln
n$ rounds.

Let $T_3$ equal $T_1 + T_2$; we set $c$ to be at least $120\ln 2$ so
that $X > 3T_3$.  We thus have $\nei{0}{2}{u}\subseteq
\nei{T_3}{1}{u}$, $\nei{0}{3}{u}\subseteq \nei{T_3}{1}{u}\cup
\nei{T_3}{2}{u}$, and $\nei{0}{4}{u}\subseteq \nei{T_3}{1}{u}\cup
\nei{T_3}{2}{u}\cup \nei{T_3}{3}{u}$. We now repeat the above analysis
again twice and obtain that at time $T = 3T_3$, $\nei{0}{2}{u}\cup
\nei{0}{3}{u}\cup \nei{0}{4}{u}\subseteq \nei{T}{1}{u}$ with
probability at least $1-\abs{\nei{0}{2}{u}\cup \nei{0}{3}{u}\cup
  \nei{0}{4}{u}}/n^3 \ge 1-1/n^2$.  By Lemma \ref{lem:moreneighbor-10p},
we have $\abs{\nei{T}{1}{u}}\ge \min\cub{2\md{0}, n-1}$, thus
completing the proof of the lemma.
\end{LabeledProof}

\begin{LabeledProof}{Lemma~\ref{lem:mingrowth-2-10p}}
Let $X$ be the first round at which $\nei{X}{2}{u} < \md{0}/4$. We
consider two cases. If $X$ is at most $cn\log n$ for a constant $c$ to
be specified later, then the claim of the lemma holds. In the
remainder of this proof we consider the case where $X$ is greater than
$cn\log n$; thus, for $0\le t\le cn\log n$, $\abs{\nei{t}{2}{u}} \ge
  \md{0}/4$. If $v\in \nei{0}{2}{u}$ is strongly tied to $\nei{0}{1}{u}$,
then
\[\prob{u\mbox{ connects to }v\mbox{ in a single round}} \ge
\frac{\dgi{t}{v}{\nei{0}{1}{u}}}{\abs{\nei{t}{1}{u}}}\cdot
\frac{1}{n}\ge \frac{\md{0}/4}{(1+1/8)\md{0}}\cdot \frac{1}{n} =
\frac{2}{9n}\] Thus, in $T = 13.5n\ln n$ rounds, $u$ will add an edge
to $v$ with probability at least $1-1/n^3$.  If there are at least
$\md{0}/8$ nodes in $\nei{0}{2}{u}$ that are strongly tied to
$\nei{0}{1}{u}$, then $u$ will add edges to all these nodes in $T$
rounds with probability at least $1-1/n^2$.

In the remainder of this proof, we focus on the case where the number
of nodes in $\nei{0}{2}{u}$ that are strongly tied to
$\nei{0}{1}{u}$ at the start of round $0$ is less than $\md{0}/8$.  In
this case, because $\abs{\nei{t}{2}{u}}\ge \md{0}/4$, more than
$\md{0}/8$ nodes in $\nei{0}{2}{u}$ are weakly tied to
$\nei{0}{1}{u}$, and, thus, have at least $3\md{0}/4$ edges to
nodes in $\nei{0}{2}{u}\cup \nei{0}{3}{u}$.

In the following we show $u$ will connect to $\md{0}/8$ nodes in
$O(n\log n)$ rounds with probability at least $1-1/n^2$.  For any
round $t$, let $W_t$ denote the set of nodes in $\nei{t}{2}{u}$
that are weakly tied to $\nei{t}{1}{u}$.  We refer to a length-2 path
from $u$ to a node two hops away as an {\em out-path}.  Let $P_0$
denote the set of out-paths to $W_0$.  Since we have at least
$\md{0}/8$ nodes in $\nei{0}{2}{u}$ that are weakly tied to
$\nei{0}{1}{u}$, $\abs{P_0}$ is at least $\md{0} /8$ at time
$t=0$. Define $e_1=\cub{u \mbox{ picks an out-path in $P_0$ and
    connects to node }v_1\mbox{ in }\nei{0}{2}{u}}$, and $X_1$ to be
the number of rounds for $e_1$ to occur.  When $0\le t\le X_1$, for
each $w_i\in \nei{t}{1}{u}$, let $f_i$ be the number of edges from
$w_i$ to nodes in $\nei{t}{1}{u}\cup \nei{t}{2}{u}$, and $p_i$ be
the number of edges from $w_i$ to nodes in $\nei{0}{2}{u}$ that are
weakly tied to $\nei{0}{1}{u}$.
\begin{eqnarray*}
\prob{e_1} 
&=& \sum_i \frac{1}{\dg{t}{u}} \cdot \frac{p_i}{f_i} 
\;\;\;\ge\;\;\; \sum_i \frac{1}{\dg{t}{u}} \cdot \frac{p_i}{n-1} 
\;\;\;=\;\;\; \frac{\sum_i p_i}{(1+1/8)\md{0}(n-1)} \\
&=& \frac{\abs{S}}{(1+1/8)\md{0}(n-1)} 
\;\;\;\ge\;\;\; \frac{\md{0}/8}{(1+1/8)\md{0}(n-1)} 
\;\;\;\ge\;\;\; \frac{1}{9n}.
\end{eqnarray*}
After $X_1$ rounds, $u$ will pick an out-path in $P_0$ and connect
such a $v_1$.  Define $P_1$ to be a set of out-paths from $u$ to
$W_{X_1}$.  We now place a lower bound on $|P_1 \setminus P_0|$.
Since $v_1\in \nei{0}{2}{u}$ is added to $\nei{X_1}{1}{u}$, those
out-paths in $P_0$ consisting of edges from $v_1$ to nodes in
$\nei{0}{1}{u}$ are not in $P_1$. The number of out-paths we lose
because of this is at most $\md{0}/4$.  But $v_1$ also has at least
$3\md{0}/4$ edges to $\nei{0}{2}{u}\cup \nei{0}{3}{u}$. The end points
of these edges are in $\nei{X_1}{1}{u}\cup \nei{X_1}{2}{u}$. If more
than $\md{0}/8$ of them are in $\nei{X_1}{1}{u}$, then $\dg{X_1}{u}\ge
(1+1/8)\md{0}$. Now let's consider the case that less than $\md{0}/8$
such end points are in $\nei{X_1}{1}{u}$. This means the number of
edges from $v_1$ to $\nei{X_1}{2}{u}$ is at least $3\md{0}/4 -
\md{0}/4 - \md{0}/8 = 3\md{0}/8$.  Among the end points of these
edges, if more than $\md{0}/8$ of them are strongly tied to
$\nei{X_1}{1}{u}$, then the degree of $u$ will become at least
$(1+1/8)\md{0}$ in $O(n \log n)$ rounds with probability $1-1/n^2$ by
our earlier argument. If not, we know that more than $\md{0}/4$ newly
added edges are pointing to nodes that are weakly tied to
$\nei{X_1}{1}{u}$.  Thus, $\abs{P_1 \setminus P_0}$ is by at least
$\md{0}/4$. $\abs{S} \ge 2\cdot \md{0}/8$. Define $e_2=\cub{u\mbox{
    picks an out-path in $P_1$ and connects to node }v_2}$, and $X_2$
to be the number of rounds for $e_2$ to occur. During time $X_1\le
t\le X_2$, $\prob{e_2}$ is at least $2\cdot\frac{1}{9n}$.  Similarly,
we define $e_3,X_3,\dots,e_{\md{0}/8},X_{\md{0}/8}$ and derive
$\prob{e_i} \ge i/(9n)$. By Lemma \ref{lem:couponcorollary-10p}, the
number of rounds for $\dg{t}{u}\ge (1+1/8)\md{0}$ is bounded by
\[T=X_1+X_2+\dots+X_{\md{0}/8} \le (2+1)9n\ln n = 27n\ln n\]
with probability at least $1-1/n^2$, completing the proof of this
lemma.
\end{LabeledProof}

\begin{LabeledProof}{Theorem~\ref{thm:graph+randwalk-10p}}
We first show that in time $T=O(n\log n)$ time, the minimum degree of
the graph increases by a factor of $1/8$, i.e., $\md{T} \ge
\min\cub{(1+1/8)\md{0}, n-1}$. Then we can apply this argument $O(\log n)$
times, and thus, complete the proof of this theorem.

For each $u$ where $\dg{0}{u} < \min\cub{(1+1/8)\md{0}, n-1}$, we
analyze by the following 2 cases. First, if $|\nei{0}{2}{u}| \ge \md{0}/2$, by Lemma
\ref{lem:mingrowth-2-10p} we know as long as $|\nei{t}{2}{u}|\ge \md{0}/4$
for all $t\ge 0$, $\dg{T}{u} \ge\min\cub{(1+1/8)\md{0},n-1}$ with
probability $1-1/n^2$ where $T=O(n\log n)$. Whenever the condition is
not satisfied, we know at least $\md{0}/4$ nodes in $\nei{0}{2}{u}$
has been moved to $\nei{T}{1}{u}$, which means $\dg{T}{u}
\ge\min\cub{(1+1/4)\md{0},n-1}$.

Second, if $|\nei{0}{2}{u}| < \md{0}/2$, by Lemma
\ref{lem:mingrowth-1-10p} we know as long as $|\nei{t}{2}{u}| < \md{0}/2$
for all $t\ge 0$, $\dg{T}{u} \ge\min\cub{(1+1/8)\md{0},n-1}$ with
probability $1-1/n^2$ where $T=O(n\log n)$. Whenever the condition is
not satisfied, we are back to the analysis in the first case, and the
minimum degree will become $\min\cub{(1+1/8)\md{0},n-1}$ with
probability $1-1/n^2$.

Combining the the above two cases we get that with
probability $1-1/n$ the minimum degree of $G$
will become at least $\min\cub{(1+1/8)\md{0},n-1}$ in $O(n\log n)$ rounds,
since there are at most $n$ nodes whose degree is between
$\md{0}$ and $\min\cub{(1+1/8)\md{0},n-1}$,

Now we can apply the above argument $O(\log n)$ times, and have shown
the two-hop walk process completes in $O(n\log^2 n)$ with high
probability. 
\end{LabeledProof}

\section{Proofs for the two-hop walk in directed graphs}
\label{app:directed}
\begin{LabeledProof}{Theorem~\ref{thm:directed.upper-10p}}
Consider any pair of nodes, $u$ and $v$.  Consider a shortest path
from $u$ to $v$ $\rb{v_0,v_1,v_2,\dots,v_m}$, where $v_0 = u$, $v_m =
v$ and $m \le n$.  Fix a time step $t$.  Let $e_i$ denote the event an
edge is added from $v_i$ to $v_{i+2}$ in step $t$. The probability of
occurrence of $e_i$ is $\prob{e_i}\ge 1/n^2$. All the $e_i$'s are
independent from one another.
\begin{eqnarray*}
\prob{\cup_i e_i} &\ge& \sum_i \prob{e_i} - \sum_i \sum_j \prob{e_i\cap e_j} \\
 &=& \sum_i \prob{e_i} - \sum_i \sum_j \prob{e_i} \prob{e_j} \\
 &\ge& m\frac{1}{n^2} - m(m-1) \frac{1}{n^4} \\
 &\ge& \frac{m}{n^2}
\end{eqnarray*}
Let $X_1$ denote the number of steps it takes for the length of the
above path to decrease by $1$.  It is clear that $E[X_1] \le n^2/m$.
In general, let $X_i$ denote the number of steps it takes for the
length of the above path to decrease by $i$.  By
Lemma~\ref{lem:couponcorollary-10p}, the number of steps it takes for
the above path to shrink to an edge is at most $4 n^2 \ln n$ with
probability $1/n^3$.  Taking a union bound over all the edges yields
the desired upper bound.

For the lower bound, consider a graph $\gf{0}$ with the node set
$\{1, 2, \ldots, n\}$ and the edge set 
\[
\{(3i,j), (3i+1,j): 0 \le i < n/4, 3n/4 \le j < n \} \bigcup \{(3i,
3i+1), (3i+1, 3i+2): 0 \le i < n/4\}.
\] 
The only edges that need to be added by the two-hop process are the
edges $(3i,3i+2)$ for $0 \le i < n/4$.  The probability that node
$3i$ adds the edge $(3i,3i+2)$ in any round is at most $16/n^2$.  The
probability that edge $(3i,3i+2)$ is not added in $(n^2 \ln n)/32$
rounds is at least $1/\sqrt{n}$.  Since the events associated with
adding each of these edges are independent, the probability that all
the $n/3$ edges are added in $(n^2 \ln n)/32$ rounds is at most $(1 -
1/\sqrt{n})^{n/3} \le e^{-\sqrt{n}/3}$, completing the lower bound
proof.
\end{LabeledProof}

\smallskip

\begin{LabeledProof}{Theorem~\ref{thm:directed.lower-10p}}
The graph $\gf{0}=(V,E)$ is depicted in Figure~\ref{fig:digraph+lower}
and formally defined as $\gf{0} = (V,E)$ where $V=\cub{1,2,\dots,n}$
with $n$ being even, and
\[
E = \cub{\rb{i,j}: 1 \le i,j \le n/2} \cup \cub{\rb{i,i+1} : n/2 \le i <
  n} \cup \cub{\rb{i,j}: i > j, i > n/2, i,j\in V}.
\]
\begin{figure}[ht]
\begin{center}
  \includegraphics[width=5in]{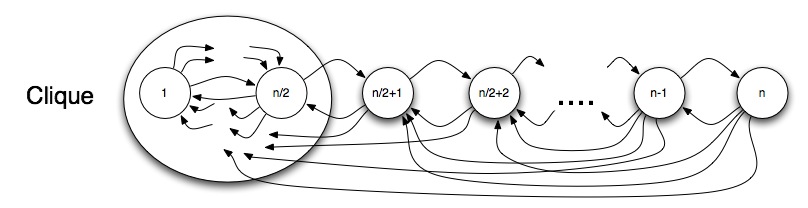}
  \caption{Lower bound example for two-hop walk process in directed graphs}
  \label{fig:digraph+lower}
\end{center}
\end{figure}
We first establish an upper bound on the probability that edge
$(i,i+h)$ is added by the start of round $t$, for given $i$, $1 \le i
\le n -h$.  Let $\prb{h}{t}$ denote this probability.  The following
base cases are immediate: $\prb{h}{0}$ is $1$ for $h = 1$ and $h < 0$,
and $0$ otherwise.  Next, the edge $(i,i+h)$ is in $\gf{t+1}$ if and
only if $(i,i+h)$ is either in $\gf{t-1}$ or added in round $t$.  In
the latter case, $(i,i+h)$ is added by a two-hop walk $i \rightarrow i
+ k \rightarrow i + h$, where $-i < k \le n - i$.  Since the
out-degree of every node is at least $n/2$, for any $k$ the
probability that $i$ takes such a walk is at most $4/n^2$.
\begin{eqnarray}
\prb{h}{t+1} & \le & \prb{h}{t} + \frac{4}{n^2} \sum_{k > -i}^{n-i} \prb{k}{t}\prb{h-k}{t} \nonumber\\
& = & \prb{h}{t} + \frac{4}{n^2} \left(\sum_{k = 1}^{i-1} \prb{h+k}{t} + \sum_{k=1}^{h-1} \prb{k}{t}\prb{h-k}{t} + \sum_{k=h+1}^{n-i} \prb{k}{t}\right) \label{eqn:lower}
\end{eqnarray}
We show by induction on $t$ that 
\begin{eqnarray}
\prb{h}{t} \le \left(\frac{\alpha t}{n^2}\right)^{h-1}, \mbox{ for all } t \le \eps n^2  \label{eqn:prb}
\end{eqnarray}
where $\alpha$ and $\eps$ are positive constants that are specified
later.

The induction base is immediate.  For the induction step, we use the
induction hypothesis for $t$ and Equation~\ref{eqn:lower} and bound
$\prb{h}{t+1}$ as follows.
\begin{eqnarray*}
\prb{h}{t+1} & \le & \left(\frac{\alpha t}{n^2}\right)^{h-1} + \frac{4}{n^2} \left(\sum_{k = 1}^{i-1} \left(\frac{\alpha t}{n^2}\right)^{h+k-1} + \sum_{k=1}^{h-1} \left(\frac{\alpha t}{n^2}\right)^{k-1} \left(\frac{\alpha t}{n^2}\right)^{h-k-1} + \sum_{k=h+1}^{n-i} \left(\frac{\alpha t}{n^2}\right)^{k-1}\right)\\
& \le & \left(\frac{\alpha t}{n^2}\right)^{h-1} + \frac{4}{n^2} \left( (h-1)\left(\frac{\alpha t}{n^2}\right)^{h-2} + \left(\frac{\alpha t}{n^2}\right)^{h} \frac{2}{1 - \alpha t/n^2}\right)\\
& \le & \left(\frac{\alpha t}{n^2}\right)^{h-1} + (h-1)\left(\frac{\alpha t}{n^2}\right)^{h-2}\frac{1}{n^2} \left(4 + \frac{4\eps^2}{(1 - \alpha\eps)}\right)\\
& \le & \left(\frac{\alpha t}{n^2}\right)^{h-1} + (h-1) \left(\frac{\alpha t}{n^2}\right)^{h-2}\frac{\alpha}{n^2}\\
& \le & \left(\frac{\alpha (t+1)}{n^2}\right)^{h-1}.
\end{eqnarray*}
(In the second inequality, we combine the first and third summations
and bound them by their infinite sums.  In the third inequality, we
use $t \le \eps n^2$.  For the fourth inequality, we set $\alpha$
sufficiently large so that $\alpha \ge 4 + 4/(1-\alpha \eps)$.  The
final inequality follows from Taylor series expansion.)

For an integer $x$, let $C_x$ denote the cut $(\{u: u \le x\}, \{v, v
> x\})$.  We say that a cut $C_x$ is {\em untouched}\/ at the start of
round $t$ if the only edge in $\gf{t}$ crossing the cut $C_x$ is the
edge $(x, x+1)$; otherwise, we say $C_x$ is {\em touched}.  Let $X$
denote the smallest integer such that $C_X$ is untouched.  We note
that $X$ is a random variable that also varies with time.  Initially,
$X = n/2$.

We divide the analysis into several phases, numbered from $0$.  A
phase ends when $X$ changes.  Let $X_i$ denote the value of $X$ at the
start of phase $i$; thus $X_0 = n/2$.  Let $T_i$ denote the number of
rounds in phase $i$.  A new edge is added to the cut $C_{X_i}$ only if
either $X_i$ selects edge $(X_i,X_i+1)$ as its first hop or a node
$u < X_i$ selects $u \rightarrow X_i \rightarrow X_i + 1$.  Since the
degree of every node is at least $n/2$, the probability that a new
edge is added to the cut $C_i$ is at most $2/n + n(4/n^2) = 6/n$,
implying that $E[T_i] \ge n/6$.

We now place a bound on $X_{i+1}$.  Fix a round $t \le \eps n^2$, and
let $E_x$ denote the event that $C_x$ is touched by round $t$.  We
first place an upper bound on the probability of $E_x$ for arbitrary
$x$ using Equation~\ref{eqn:prb}.
\[
\Pr[E_x] \le \sum_{h \ge 2} h \left(\frac{\alpha t}{n^2} \right)^{h-1} \le \frac{\alpha t (4 - 3(\alpha t)/n^2 + (\alpha t)^2/n^4)}{n^2(1-(\alpha t)/n^2)^3},
\]
for $t \le \eps n^2$, where we use the inequality $\sum_{h \ge 2} h^2
\delta^h = \delta(4-3\delta+\delta^2)/(1-\delta)^3$ for $0 < \delta <
1$.  We set $\eps$ sufficiently small so that
$(4-3\eps+\eps^2)/(1-\eps)^3 \le 5$, implying that the above
probability is at most $5\eps$.

If $E_x$ were independent from $E_y$ for $x \neq y$, then we can
invoke a straightforward analysis using a geometric probability
distribution to argue that $E[X_{i+1} - X_i]$ is at most $1/(1-5\eps)
= O(1)$; to see this, observe that $X_{i+1} - X_i$ is stochastically
dominated by the number of tosses of a biased coin needed to get one
head, where the probability of tail is $5\eps$.  The preceding
independence does not hold, however; in fact, for $y > x$, $\Pr[E_y
  \mod E_x] > \Pr[E_y]$.  We show that the impact of this correlation
is very small when $x$ and $y$ are sufficiently far apart.  We
consider a sequence of cuts $C_{x_1}, C_{x_2}, \ldots, C_{x_\ell},
\ldots$ where $x_\ell = x_{\ell-1} + c\ell$, for a constant $c$ chosen
sufficiently large, and we set $x_0 = X_i + 2$.  We bound the
conditional probability of $E_{x_\ell}$ given $E_{x_{\ell -1}} \cap
E_{x_{\ell-2}} \cdots \cap E_{x_1}$ as follows.
\begin{eqnarray*}
& & \Pr[E_{x_\ell} | E_{x_{\ell -1}} \cap E_{x_{\ell-2}} \cdots E_{x_1}]\\
& = & \frac{\Pr[E_{x_\ell} \cap E_{x_{\ell -1}} \cap E_{x_{\ell-2}} \cdots E_{x_1}]}{\Pr[E_{x_{\ell -1}} \cap E_{x_{\ell-2}} \cdots E_{x_1}]}\\
& \le & \frac{\Pr[E_{x_{\ell -1}} \cap E_{x_{\ell-2}} \cdots E_{x_1} \cap (C_{x_\ell} \cap (C_{x_{\ell-1}} \cup \cdots \cup C_{x_1}) = \emptyset)]}{\Pr[E_{x_{\ell -1}} \cap E_{x_{\ell-2}} \cdots E_{x_1}]} +\\
& & 
\frac{\Pr[E_{x_{\ell -1}} \cap E_{x_{\ell-2}} \cdots E_{x_1} \cap (C_{x_\ell} \cap (C_{x_{\ell-1}} \cup \cdots \cup C_{x_1}) \neq \emptyset)}{\Pr[E_{x_{\ell -1}} \cap E_{x_{\ell-2}} \cdots E_{x_1}]}\\
& \le & \frac{\Pr[E_{x_{\ell -1}} \cap E_{x_{\ell-2}} \cdots E_{x_1}] \Pr[\mbox{a new edge is added from } (x_{\ell-1} +1, x_\ell) \mbox{ to } (x_\ell+1,n]]}{\Pr[E_{x_{\ell -1}} \cap E_{x_{\ell-2}} \cdots E_{x_1}]}\\
& & \frac{\Pr[\mbox{an edge spanning at least } c\ell \mbox{ hops is added across } C_{x_\ell}]}{\Pr[E_{x_{\ell -1}} \cap E_{x_{\ell-2}} \cdots E_{x_1}]}\\
& \le & \Pr[E_{x_\ell}] + \frac{((\alpha t)/n^2)^{c\ell-1}}{(1 - \alpha t/n^2)^2(t/n^2)^{\ell}}\\
& \le & 5\eps + \eps = 6\eps,
\end{eqnarray*}
where we set $c$ sufficiently large in the last step.  Since $X_{i+1}$
is at most the smallest $x_\ell$ such that $C_{x_\ell}$ is untouched, 
we obtain that
\begin{eqnarray*}
E[X_{i+1} - X_i] \le 2 + \sum_{\ell \ge 2} (6 \eps)^\ell c\ell^2 \le c',
\end{eqnarray*}
for a constant $c'$ chosen sufficiently large.  We thus obtain that
after $\eps' n$ phases, $E[X]$ is at most $\eps'c' n$, where $\eps'$
is chosen sufficiently small so that $n - E[X]$ is $\Omega(n)$.  Since
the expected length of each phase is at least $n/6$, it follows that
the expected number of rounds it takes for the two-hop process to
complete is $\Omega(n^2)$ rounds.

\junk{Let $e_i$ be the event of
  adding an edge from node $i$ to node $i+2$. $\prob{e_i}\le 4/n^2$
  $\forall i\ge n/2$. And all $e_i$'s are independent. We will show it
  takes $\Omega(n^{2})$ time for all $e_i$'s ($i\ge n/2$) to
  happen. Let $r_i$ be the event of adding an edge from node $i$ to
  node $i+3$. $\forall i$, the probability that $r_i$ happens in
  $n^{2}/4$ rounds is
\begin{eqnarray*}
\prob{r_i\mbox{ in }n^{2}/4\mbox{ rounds}} &=& {n^{2}/4 \choose 2}\prob{e_i}\prob{r_i|e_i} \\
&=& \frac{n^2/4(n^2/4-1)}{2} \frac{4}{n^2} \frac{4}{n^2} \\
&\le& \frac{1}{2}
\end{eqnarray*}
Therefore,
\[\prob{\mbox{more than }\ln n\mbox{ } r_i\mbox{'s happen in }n^2/4\mbox{ rounds}} \le \frac{1}{n}\]
\begin{figure}[ht]
\begin{center}
  \includegraphics[width=3in]{diexample2.jpg}
  \caption{Bad events of 2-step random walk in directed graph}
  \label{fig:digraph+lower+badevent}
\end{center}
\end{figure}
As shown in Figure \ref{fig:digraph+lower+badevent}, such $r_v$ event
will increase the probability of $e_{v+1}$. Since such bad event won't
be more than $\ln n$ with probability at least $1-1/n$, we can look at
the other $n/2-\ln n$ $e_i$'s. In $n^2/4$ rounds, with constant
probability not all $e_i$'s will complete. Thus, the lower bound for
2-hop random walk process in directed graph is $\Omega(n^2)$.
}

\end{LabeledProof}

\end{document}